\documentclass[11pt]{article}
\usepackage[left=1in, right=1in, top=1in, bottom=1in]{geometry}

\usepackage{tabularx,ragged2e,multirow}
\newcolumntype{C}{>{\Centering}X}
\usepackage{booktabs}

\usepackage{graphicx} 
\usepackage{float} 
\usepackage{enumitem} 
\usepackage{verbatim}
\usepackage[most]{tcolorbox} 
\usepackage[T1]{fontenc}
\usepackage[section]{placeins}

\usepackage{amsmath, amsthm, amssymb, amsfonts, mathtools}

\newtheorem{theorem}{Theorem}[section]
\newtheorem{definition}[theorem]{Definition}
\newtheorem{proposition}[theorem]{Proposition}
\newtheorem{lemma}[theorem]{Lemma}

\usepackage{bm}

\newcommand{\abstain}{\bot}

\DeclareMathOperator*{\argmax}{\mathrm{argmax}}

\newcommand{\bsc}{\mathbf{c}}

\usepackage[normalem]{ulem} 
\usepackage[numbers,sort&compress]{natbib}

\usepackage{hyperref}
\usepackage{xurl} 
\usepackage[nameinlink]{cleveref} 

\title{Audited Auctions: Reducing Harms in Advertising}

\author{%
\begin{tabular}{c@{\hspace{40pt}}c}
Tejovan Parker\thanks{Center for Computing and Data Sciences, Boston University. \href{mailto:\{tejovanp,gmaayan\}@bu.edu}{\{tejovanp,gmaayan\}@bu.edu}} &
Gabriel McDonnell Maayan\footnotemark[1] \\[8pt]
Francisco Marmolejo-Cossío\thanks{Department of Computer Science, Boston College. Berkman Klein Center and School of Engineering and Applied Sciences, Harvard University. Input Output Global (IOG). \href{mailto:marmolf@bc.edu}{marmolf@bc.edu}} &
Marshall Van Alstyne\thanks{Questrom School of Business, Boston University. Sloan School of Management, Massachusetts Institute of Technology. \href{mailto:mva@bu.edu}{mva@bu.edu}}
\end{tabular}%
}
\date{July 17, 2026}

\begin{document}
\maketitle
\begin{abstract}
\footnotesize
Although standard auction mechanisms help truthfully reveal preferences of \emph{bidders}, they can inadvertently result in unbounded harms when they fail to account for externalities caused by bid allocations affecting \textit{non-bidders}. Ad markets, that buy and sell \emph{user} attention represent such auctions.  This research explores a welfare improving auctioneer's audit-and-penalty mechanism that helps screen the worst externalities. We prove this mechanism can internalize externalities formally, then explore social welfare gains empirically.

\end{abstract}
\tableofcontents

\section{Introduction}

Real time bid (RTB) ad auctions represented at least a \$14 billion market in 2024, with growth expected to exceed 11\% through 2030.\footnote{Of five market research firm estimates, we use the most conservative. The lowest was \$14.4B with growth of 16\% (\href{https://www.grandviewresearch.com/industry-analysis/real-time-bidding-market-report}{Grand View Research}), while the highest was \$21.4B with growth of 11\% (\href{https://www.marketresearchfuture.com/reports/real-time-bidding-market-7674}{Market Research Future}).} As a transaction between an ad buyer and an ad seller, these produce externalities in the sale of third party user attention. Standard auction mechanisms that optimize disproportionately for buyer-seller value do not optimize for total social welfare. 
This is one reason social media platforms famously prioritize harmful, misleading, engaging content that places ``profits over people''\citep{georgia_wells_is_2021}. 
This research seeks to re-balance auction optimization to account for externalities.

We study simple auctions -- identical-item, unit-allocation -- in the presence of probabilistic penalties that depend on externalities which we refer to as `audits'. Ad buyers, who participate absent audits, exit the market when their social harms are large enough.  We show formally, and with computational simulation, when they improve social welfare. Combining data from X-Twitter with third party research on advertising, we simulate auctions and roughly estimate welfare gains from audits on the order of \$1-10 million per month from political advertisements on one major social media platform. This estimate is illustrative, and is not accurate enough to be a reliable quantitative forecast.

Of theoretical and practical interest, our mechanism(s) improve on alternative externality solutions that are infeasible based on transaction costs and information asymmetries.  Pricing externalities might normally occur if those affected by harms could also place bids in ad auctions. That approach not only disadvantages under-resourced populations, but also fails completely in attention markets based on a straightforward ``inspection paradox.'' 
To judge whether a message is worthy of attention, a recipient must attend to the message but, upon judging it \emph{un}worthy, recipients cannot reclaim their attention.  Further, the number of ways any given individual might be affected by any other two-party transactions is arbitrarily large, placing undue burden on those experiencing harms.


Our mechanism instead places greater responsibility on speakers promoting content. If a promoted message generates sufficiently harmful externalities, speakers may incur penalties tied to those harms. This shifts part of the burden of harmful information environments away from affected users and toward content promoters themselves.

\subsubsection*{Contributions and Paper Outline}

Our contributions are threefold. First, we introduce auctions with audits, a mechanism-design framework in which externality-dependent penalties induce endogenous participation constraints. Second, we characterize structural properties of participant and recipient audit mechanisms, including revelation-principle-style reductions and welfare-maximizing implementations under VCG. Third, using data derived from X Community Notes, we simulate the effects of audit-based mechanisms in attention-market environments.

This paper proceeds as follows. The next section \ref{sec:motivation_lit_review} reviews related literature and motivates our work, including a toy illustration of how standard methods like VCG auctions struggle to incorporate complete welfare and may lead to arbitrarily negative externalities. We follow with a formal model and proofs of our main propositions in sections \ref{sec:model}-\ref{sec:recipient-audit}, while discussing limitations and dangers of audits in section \ref{sec:dangers_of_audits}. In Section \ref{sec:numerical-empirical} we characterize welfare optimal audits in more general settings using numerical simulations. In section \ref{sec:data-collection} we also construct an empirical dataset from X-Twitter posts and community notes, and roughly estimate real social welfare gains. Our main contribution is to offer a novel solution to a significant social problem, motivating audits as a method for improving welfare in auctions with externalities and imperfect information. Our results can also be used to understand natural processes which function like audits and penalties. We end with discussion and limitations in section \ref{sec:conclusion}. Additional materials are available in the \hyperref[appendix:optimization]{Appendix}. 

\section{Motivation and Related Work}\label{sec:motivation_lit_review}

\subsection*{Toy Illustration of Unbounded Harm from VCG}

Consider an auction allocating a single indivisible item among many bidders with heterogeneous private valuations. The auctioneer uses a VCG mechanism: the highest-valuation bidder wins and pays the second-highest bid, truthfully eliciting preferences and maximizing bidder welfare.

We augment the standard welfare term $v_i$ with an externality $e_i$, so total welfare depends on both the winning bidder's private valuation and the broader effects of allocation. We define $v_i$ to include bidder $i$'s direct utility and any ``other-regarding preferences'' \citep{fehr1999theory} already internalized, such as altruism or fairness concerns. The term $e_i$ captures effects \emph{not} internalized in the bidder's decision calculus and, absent constraints, may become arbitrarily large in magnitude. Thus, the welfare contribution $(v_i+e_i)$ of the highest bidder can be arbitrarily negative, leading to arbitrarily bad welfare.

Examples provide useful intuition. These might include ads for (i) fossil fuel vehicles that do not account for global warming, (ii) fast fashion that do not account for textile and chemical waste, (iii) disposable consumables that do not account for landfill and microplastics, (iv) politics that do not account for polarization and institutional dysfunction, or (iv) scam medications that do not account for ill health and failed herd immunity. These exhibit \emph{production} externalities, e.g. chemical waste from manufacture, \emph{consumption} externalities, e.g. carbon emissions from driving, and even \emph{attention} externalities, e.g. focused outrage displacing focused learning or family time. Harms are delayed, diffused, and borne unequally among third parties. 

\subsection{Externalities and Information Asymmetry in Attention Markets}

Modern attention markets allocate visibility through auction mechanisms, including sponsored-search and advertising auctions \cite{edelman_internet_2007, aggarwal_truthful_2006, varian_position_2007}. Although practical systems involve quality adjustments, real-time bidding, budgets, and heterogeneous display positions, their core strategic properties can be studied through simpler multi-item models. We use a unit-demand multiple-identical-item auction as a tractable abstraction.

Classical auction mechanisms primarily optimize advertiser utility, platform revenue, or bidder welfare. However, promoted content may also generate broader externalities affecting listeners, users, or society more generally. These externalities may be positive (e.g.\ educational or socially beneficial information) or negative (e.g.\ misinformation, manipulation, spam, or harmful advertising). As digital platforms increasingly mediate economic, political, and social interactions, these externalities become an increasingly important component of welfare in attention markets \cite{vettehen_attention_2023}.

Mechanisms such as VCG can maximize total welfare when all affected parties values are reported to the mechanism \cite{vickrey_counterspeculation_1961, clarke_multipart_1971, groves_incentives_1973}. Related approaches incorporate social costs into sponsored-search auctions \cite{abrams_ad_2007} and automated bidding \cite{deng_autobidding_2023}. In practice, however, centralized estimation of information externalities is costly, and platforms have repeatedly demonstrated willingness to maximize engagement at welfare's expense \cite{vettehen_attention_2023}.

Users could in principle report their own preferences over content. The core difficulty is informational: the value of information is often uncertain before consumption and may remain unknown even afterward—a feature characteristic of credence goods \cite{emons_credence_1997, loewenstein_economics_2025}. Traditional quality signals such as speaker reputation or institutional affiliation have become unreliable in high-volume environments where bad actors strategically ``flood the zone'' \cite{vettehen_attention_2023}. Even users formally empowered to influence allocation may therefore lack the information or attention to do so meaningfully.

Eliciting social value from advertiser and/or listener reports corresponds naturally with the interdependent-valuation setting \cite{milgrom_theory_1982, eden_private_2021, eden_constant_2023, eden_private_2024}. However, speakers face strong incentives to promote content regardless of its social value and there may be large information asymmetries. Mechanisms for the interdependent-valuation settings are limited in their ability to handle these issues.

These difficulties suggest that directly pricing externalities in attention markets is impractical, motivating mechanisms that instead shape participation incentives indirectly—without requiring centralized knowledge of externalities or complete preference elicitation from affected parties.

Our approach is inspired by several related ideas from costly signaling, verification, and decentralized speech regulation. A recent line of work proposes voluntary liability systems in which speakers may choose to stake resources or accept penalties contingent on later adjudication of their speech \cite{van_alstyne_free_2021, van_alstyne_free_2023, van_alstyne_consequentialist_2023, arbel_truth_2024, lin_towards_2023, van_alstyne_improving_2023}. These proposals are motivated by the possibility of creating decentralized alternatives to both unrestricted information environments and centralized moderation regimes. Earlier work studied related ideas in the context of spam communication and attention allocation \cite{loder_economic_2006}. Our model can be interpreted as a simplified mechanism-design abstraction inspired by these approaches.

\subsection{Verification, Costly Signaling and Entry}

More broadly, a large literature studies mechanism design with verification, evidence, auditing, and costly signaling. Foundational work on costly state verification considers settings in which a principal may pay to verify information held by strategic agents \cite{townsend_optimal_1979, gale_incentive-compatible_1985, border_samurai_1987, mookherjee_optimal_1989}. More recent work studies verification in allocation settings, including ex-post verification, noisy verification, and verification under distributional ambiguity \cite{ben-porath_optimal_2014, li_mechanism_2020, epitropou_optimal_2019, hu_screening_2024, bayrak_distributionally_2025}. Related literatures study evidence disclosure, partial verification, lie detection, statistical testing, and costly signaling \cite{green_partially_1986, ben-porath_mechanisms_2019, fotakis_mechanism_2016, bates_incentive-theoretic_2024, spence_job_1973, hartline_optimal_2008}.

Our model differs from these approaches in several important ways. First, audits occur after allocation and do not directly affect the allocation rule itself. Instead, audits shape incentives indirectly through participation decisions. Second, penalties depend on bidder externalities rather than solely on misreports or the auctioneer's own valuation. Finally, our model is closely related to auctions with participation costs or entry fees \cite{samuelson_competitive_1985, menezes_auctions_2000, cao_equilibria_2018}, but differs in that the participation cost is both endogenous to bidder externalities and controllable by the mechanism designer.

Conceptually, our approach uses audits to induce endogenous participation constraints in otherwise standard auction mechanisms. In participant-audit auctions, audits affect incentives to participate while preserving truthful bidding incentives conditional on participation. This leads to a revelation-principle-style reduction. In recipient-audit auctions, audits transform bidder values in a simple way, and in principle can obtain perfect welfare maximization.

\section{Participant Audit Model}\label{sec:model}

We consider auctions for $k \geq 1$ identical goods where $n$ bidders compete to each obtain at most one of the $k$ goods for sale. We identify the set of bidders with $[n] = \{1,\dots,n\}$. Moreover, we assume that each bidder has a type represented by two parameters: their value for the good, $v_i \in \mathbb{R}_{\geq 0}$, and an externality for them obtaining the good, $e_i \in \mathbb{R}$. 

For each bidder, we assume that $(e_i,v_i)$ is drawn according to a publicly known distribution with CDF $H_i$, and that once drawn, $(e_i,v_i)$ is private information.  However, we also assume that the auctioneer has the ability to {\it audit} bidders to learn their true externality after goods have been allocated. With this added capability for the auctioneer, we can explore mechanisms that penalize bidders {\it as a function of their externality}. Importantly, these penalties operate separately from the underlying allocation and payment rules of the auction itself, affecting bidder participation incentives rather than directly modifying bids.

\subsection{Auction Format}

In the auctions we consider, we assume that bidders make bids with respect to their value for the good, $v_i$. Externalities are not reported by bidders. In what follows, suppose that the $i$-th bidder's bid is given by $b_i \in B$, where the set $B =\mathbb{R}_{\geq 0} \cup \{\abstain\}$ is the space of potential bids. We use the notation $b = (b_1,\dots,b_n) \in B^n$ for joint bid profiles and note that a bid of $b_i = \abstain$ represents a bidder abstaining from the auction. 

\begin{definition}[Auction with Audits]
An {\it Auction with Audits} is specified by a tuple:
$$
\mathcal{A} = (x,p,r)
$$
Where each component is as follows:
\begin{itemize}
    \item $x: B^n \rightarrow \{0,1\}^n$ is the allocation function, where we use the shorthand $x_i(b) = x(b)_i \in \{0,1\}$ to represent whether the $i$-th bidder receives a good when joint bids $b \in B^n$ are made. We require that $x_i(\abstain,b_{-i}) = 0$ for any $b_{-i} \in B_{-i}$ and that $\sum_{i=1}^n x_i(b) \leq k$. 
    \item $p: B^n \rightarrow \mathbb{R}^n$ is the payment function, where we use the shorthand $p_i(b) = p(b)_i \in \mathbb{R}$ to represent the amount paid by the $i$-th bidder when joint bids $b \in B^n$ are made. We require that $p_i(\bot,b_{-i}) = 0$ for any $b_{-i} \in B_{-i}$. 
    \item $r: \mathbb{R} \rightarrow \mathbb{R}$ is a penalty function, such that a bidder with externality $e_i$ pays a penalty given by $r(e_i)$ when they participate in the auction (i.e. bid $b_i \neq \abstain$). 
\end{itemize}
\end{definition}

\subsection{Bidder Incentives and Equilibria}

\paragraph{Bidder utilities.}
We assume that bidder values for each of the indistinguishable goods are independent of other bidder values, and that each bidder demands at most one of the $k$ goods. Moreover, we assume bidders have quasi-linear utilities when not abstaining, given by:
$$
u_i(b) = v_i \cdot x_i(b) - p_i(b) - r(e_i) \cdot \mathbb{I}(b_i \neq \abstain),  
$$
where $\mathbb{I}(\cdot)$ is an indicator function. We interpret $r(e_i)$ as the bidder's {\it expected} audit penalty conditional on participation.\footnote{More concretely, suppose that the auctioneer commits to randomly auditing any given bidder with $b_i \neq \abstain$ with probability $\pi \in (0,1)$. If the auctioneer wanted to impose a penalty of $r(e_i)$ in expectation, they could instead impose a penalty of $\frac{1}{\pi}r(e_i)$ conditional on audit, resulting in expected penalty $r(e_i)$ as desired.}
Importantly, the penalty term depends only on participation and the bidder's externality, and not directly on the submitted bid itself. As a result, externalities affect incentives to participate in the auction, while conditional bidding incentives are determined entirely by the underlying auction mechanism. This separation will play a central role in the structural results that follow. Finally, as we are in a Bayesian setting, our equilibrium concepts will ultimately concern expected utilities at equilibrium.

\paragraph{Bidder strategies.} We consider bidder strategies given by $s_i: \mathbb{R}^2 \rightarrow B$, where $s_i(e_i,v_i) = b_i$ represents the bid made by the $i$-th player.

\begin{definition}[Bayes-Nash Equilibrium (BNE)] We recall that the $i$-th player's type is independently drawn from $H_i$. Consequently, we say that $s^* = (s^*_1,\dots,s^*_n)$ is a {\it Bayes-Nash Equilibrium} for auction $\mathcal{A} = (x,p,r)$ if and only if the following holds for each bidder, $i \in N$ with type $(e_i,v_i) \in \mathbb{R}^2$, and each potential bid $b_i' \in B$:
$$
\mathbb{E}_{(e_{-i},v_{-i}) \sim H_{-i}} \left[ u_i(s^*_i(e_i,v_i), s^*_{-i}(e_{-i},v_{-i})) \right] \geq 
\mathbb{E}_{(e_{-i},v_{-i}) \sim H_{-i}} \left[ u_i(b_i', s^*_{-i}(e_{-i},v_{-i})) \right]
$$
\end{definition}

Auctions with audits are similar to auctions with entry fees, with the important distinction that penalties may depend on bidder externalities. As a result, the auctioneer can indirectly influence which bidder types choose to participate in the auction. Because participation itself becomes endogenous, truthful bidding is only meaningful conditional on participation. This motivates the following equilibrium refinement.

\begin{definition}[Participant Truthful Bayes-Nash Equilibrium (PT-BNE)]
Suppose that $s^*=(s_1^*,\dots,s_n^*)$ is a Bayes-Nash Equilibrium of an auction with audits $\mathcal{A}$. We say that $s^*$ is a \emph{Participant Truthful Bayes-Nash Equilibrium (PT-BNE)} if for every bidder type $(e_i,v_i)$,
\[
s_i^*(e_i,v_i)\neq \abstain
\quad\Longrightarrow\quad
s_i^*(e_i,v_i)=v_i.
\]
\end{definition}

Having defined BNE, we specify relevant notions of implementability for auctions with audits. 

\begin{definition}[Bayes-Nash Implementable with Audits]
Suppose that $s^* = (s_1^*,\dots,s_n^*)$ is a joint strategy of all bidders with types drawn independently from distributions $H_1,\dots,H_n$. We say that $s^*$ is Bayes-Nash Implementable with Audits (BNI-A) if there exists an auction with audits, $\mathcal{A} = (x,p,r)$ such that $s^*$ is a BNE of $\mathcal{A}$.  
\end{definition}

\begin{definition}[Participant Truthful Auction]
Consider a population of $n$ bidders with types drawn from independent distributions $H_1,\dots,H_n$. We say that an auction with audits, $\mathcal{A} = (x,p,r)$ is {\it Participant Truthful (PT)} if it has a PT-BNE. 
\end{definition}

\subsection{Objectives}

Game theoretically, we are interested in auctions with audits that are PT. This is due to a revelation principle for auctions with audits, which shows that if $s^*$ is BNI-A via auction $\mathcal{A}$, then there exists a PT auction, $\mathcal{A}'$ with identical outcomes to $\mathcal{A}$ at a PT-BNE. We will provide a proof of this revelation principle in the following section, but we note that this allows us, without loss of generality, to restrict attention to PT auctions with audits.

Throughout this paper, we will also assume that the auctioneer seeks to optimize the value of bidders obtaining goods as well as the externality of the given allocation of goods. As such we define welfare as follows:
\begin{definition}[Welfare]
Define an individual bidder $i$'s welfare contribution as $w_i(e_i,v_i) = v_i + e_i$. For a given auction $\mathcal{A}$, we say that the welfare achieved with joint bids $b = (b_1\dots,b_n)$ is given by:\footnote{It is straightforward to extend our results for participant and recipient audits to welfare functions of the form $w_i=v_i+f(e_i)$ for any function $f$.}
$$
W_{\mathcal{A}}(b) = \sum_{i=1}^n x_i(b) w_i(v_i, e_i).
$$ 
\end{definition}
When clear from context we omit the subscript from welfare.\footnote{Notice that a welfare-maximizing allocation of goods may not end up allocating any goods at all if total externalities are far more negative than the individual value of goods for bidders.}

\begin{tcolorbox}[colframe=black,colback=gray!10,boxrule=0.5pt,arc=3pt]
\begin{center}
{\bf Our objective:}\\
Describe PT auctions with audits that maximize welfare, improving on prior results.
\end{center}
\end{tcolorbox}

\section{Structural Results for PT Auctions with Audits}

In this section we prove structural results for PT auctions with audits. Our main observation is that externality penalties affect bidder participation separately from the bidding behavior used by the auction mechanism itself. Conditional on participation, equilibrium bid behavior depends only on bidder values and not on externalities. This separation yields a revelation-principle-style reduction showing that it suffices to study PT auctions. It also gives a simple recipe for constructing PT auctions with audits from classical BNIC auctions that ignore bidder externalities.

\subsection{Revelation Principle for Auctions with Audits}

We begin by showing that, without loss of generality, equilibria of auctions with audits can be represented via PT auctions.
Intuitively, because externality penalties only affect participation decisions and not conditional bidding behavior, any equilibrium behavior can be simulated by a truthful direct mechanism that separately handles participation and allocation.

\begin{theorem}[Revelation Principle for Auctions with Audits]
\label{thm:revelation-auction-audit}

Suppose that $s^*=(s_1^*,\dots,s_n^*)$ is BNI-A via auction
$\mathcal{A}=(x,p,r)$. Then there exists an auction
$\mathcal{A}'=(x',p',r')$ with a PT-BNE
$\bar s=(\bar s_1,\dots,\bar s_n)$ such that for every type profile
$\{(e_i,v_i)\}_{i=1}^n$:

\begin{itemize}
    \item $s_i^*(e_i,v_i)=\abstain \iff \bar s_i(e_i,v_i)=\abstain$;

    \item the induced allocation distribution under $\bar s$ in $\mathcal{A}'$
    is identical to that induced by $s^*$ in $\mathcal{A}$;

    \item the induced payment distribution under $\bar s$ in $\mathcal{A}'$
    is identical to that induced by $s^*$ in $\mathcal{A}$.
\end{itemize}

\end{theorem}

Before we continue to the proof of Theorem \ref{thm:revelation-auction-audit}, we prove a key lemma:
\begin{lemma}[Best-response structure]
\label{lemma:BNE-structure}
Fix any joint strategy profile $s^*_{-i}$ for the other bidders. For any bidder $i$ and any value $v_i$, the set of optimal non-abstention bids for bidder $i$ is independent of the externality $e_i$. In particular, if two types $(e_i,v_i)$ and $(e'_i,v_i)$ both participate, then they have the same set of best responses among bids in $\mathbb{R}_{\geq 0}$.
\end{lemma}

\begin{proof}
Fix $s^*_{-i}$ and consider bidder $i$ with type $(e_i,v_i)$. For any non-abstention bid $b_i \in \mathbb{R}_{\geq 0}$, bidder $i$'s interim utility is
\[
U_i(b_i;e_i,v_i)
=
\mathbb{E}_{(e_{-i},v_{-i})\sim H_{-i}}
\left[
v_i x_i(b_i,s^*_{-i}(e_{-i},v_{-i}))
-
p_i(b_i,s^*_{-i}(e_{-i},v_{-i}))
\right]
-
r(e_i).
\]
The final term, $r(e_i)$, does not depend on the bid $b_i$. Therefore, for any two non-abstention bids $b_i,b'_i \in \mathbb{R}_{\geq 0}$,
\[
U_i(b_i;e_i,v_i) \geq U_i(b'_i;e_i,v_i)
\]
if and only if
\[
U_i(b_i;e'_i,v_i) \geq U_i(b'_i;e'_i,v_i).
\]
Thus the ranking of all non-abstention bids is independent of $e_i$, and so the set of optimal non-abstention bids depends only on $v_i$ and the fixed strategy profile $s^*_{-i}$.
\end{proof}

The structural property established in Lemma~\ref{lemma:BNE-structure} implies that, conditional on participation, equilibrium bid behavior depends only on bidder values and not on externalities. Thus, the auction component need not observe or elicit externalities in order to simulate the allocation and payment outcomes of $s^*$; externalities enter only through the separate penalty/participation channel.

\begin{proof}
Suppose $\mathcal{A}=(x,p,r)$ has BNE $s^*=(s_1^*,\dots,s_n^*)$. For each value $v_i$, let
\[
B^*(v_i)
=
\{b_i \in \mathbb{R}_{\geq 0} : b_i \text{ is an optimal non-abstention bid for } v_i\}.
\]
By Lemma~\ref{lemma:BNE-structure}, $B^*(v_i)$ is independent of $e_i$. We construct an auction $\mathcal{A}'=(x',p',r')$ as follows. Given reports $\bar b=(\bar b_1,\dots,\bar b_n)$, define simulated bids $b=(b_1,\dots,b_n)$ independently for each bidder:
\begin{itemize}
    \item if $\bar b_i=\abstain$, set $b_i=\abstain$;
    \item otherwise sample $b_i \in B^*(\bar b_i)$ according to the conditional distribution induced by $s_i^*(e_i,v_i)$ given $v_i=\bar b_i$ and participation.
    Importantly, this simulation uses only the reported value $\bar b_i$ and the induced distribution of non-abstention bids conditional on that value; it does not require $\mathcal{A}'$ to observe $e_i$ when computing allocations or payments.
\end{itemize}
The auction $\mathcal{A}'$ then runs $\mathcal{A}$ on the simulated bid profile $b$. Finally, define
\[
\bar s_i(e_i,v_i)
=
\begin{cases}
v_i & \text{if } s_i^*(e_i,v_i)\neq \abstain,\\
\abstain & \text{otherwise.}
\end{cases}
\]
Under $\bar s$, the simulated bids reproduce the same distribution over bids, allocations, and payments as $s^*$ in $\mathcal{A}$. Since deviations in $\mathcal{A}'$ induce equivalent deviations in $\mathcal{A}$, and $s^*$ is a BNE of $\mathcal{A}$, it follows that $\bar s$ is a PT-BNE of $\mathcal{A}'$.
\end{proof}

\subsection{Creating PT Auctions with Audits from BNIC Auctions}

The structural separation between participation incentives and conditional bidding behavior yields a particularly simple method for constructing PT auctions with audits. In particular, any classical BNIC auction can be converted into a PT auction by introducing participation penalties that selectively discourage bidders with sufficiently negative externalities. To formalize this connection, we first recall the standard notion of auctions without audits.

\paragraph{Classical Auctions without Audits. }In what follows, we also consider auctions in a classical Bayesian setting where bidders have independent valuations. As before, we assume $k$ indistinguishable goods are being sold by an auctioneer to $n$ bidders, where each bidder has a private value $v_i \in \mathbb{R}_{\geq 0}$ for a good, and demands at most one good. Finally, in this classical setting, we assume that valuations are independently distributed with $v_i \sim F_i$ for a CDF $F_i$ for each bidder. 
We assume that the $i$-th bidder can make bids $b_i \in \mathbb{R}_{\geq 0}$, and that a bidding strategy for said bidder corresponds to a function $s_i: \mathbb{R}_{\geq 0} \rightarrow \mathbb{R}_{\geq 0}$ such that $b_i = s_i(v_i)$. 

An auction without audits is a tuple $\mathcal{A} = (x,p)$ consisting of an allocation rule $x: (\mathbb{R}_{\geq 0})^n \rightarrow \{0,1\}^n$, and a payment rule $p: (\mathbb{R}_{\geq 0})^n \rightarrow \mathbb{R}^n$, where $x_i(b) \in \{0,1\}^n$ is the allocation to the $i$-th bidder and $p_i(b) \in \mathbb{R}$ is the payment for the $i$-th bidder. Furthermore, $\sum_{i=1}^nx_i(b) \leq k$ for all $b \in (\mathbb{R}_{\geq 0})^n$. Finally, we assume that bidders have quasi-linear utilities given by $u_i(b) = v_i \cdot x_i(b) - p_i(b)$.  

Given an auction $\mathcal{A} = (x,p)$, we say that $s^* = (s^*_1,\dots,s^*_n)$ is a Bayes-Nash Equilibrium if and only if the following holds for any bidder $i \in N$, any value $v_i \in \mathbb{R}_{\geq 0}$ and any bid $b_i' \in \mathbb{R}_{\geq 0}$
$$
\mathbb{E}_{v_{-i} \sim F_{-i}} \left[ u_i(s^*_i(v_i), s^*_{-i}(v_{-i})) \right] \geq 
\mathbb{E}_{v_{-i} \sim F_{-i}} \left[ u_i(b_i', s^*_{-i}(v_{-i})) \right].
$$
Finally, we say that an auction $\mathcal{A} = (x,p)$ is Bayes-Nash Incentive Compatible (BNIC) if $s^* = (s^*_1,\dots,s^*_n)$ such that $s^*_i(v_i) = v_i$ is a BNE of $\mathcal{A}$. With this in hand, we are in a position to 

\begin{theorem}
\label{thm:BNIC-to-PT}
Consider $n$ bidders with types given by $(e_i,v_i) \sim H_i$, such that the marginal distribution of values is given by $v_i \sim F_i$. Suppose that $\mathcal{A} = (x,p)$ is an auction (without audits) which is BNIC for the Bayesian setting given by $v_i \sim F_i$ independently for each bidder. Then for any penalty function $r:\mathbb{R}\to\mathbb{R}$, there exists a PT auction with audits $\mathcal{A}'=(x,p,r)$ that preserves truthful bidding conditional on participation.
\end{theorem}

Intuitively, because truthful bidding is already incentive compatible in the underlying auction, the penalty function only determines which bidder types prefer to participate.

\begin{proof}

Suppose that $s^* = (s^*_1,\dots,s^*_n)$ is a truthful BNE of $\mathcal{A} = (x,p)$ in the non-audit setting and that we are given a penalty function $r: \mathbb{R} \rightarrow \mathbb{R}$ to define the auction with audit $\mathcal{A}' = (x,p,r)$. We define a strategy profile, $\bar{s} = (\bar{s}_1,\dots,\bar{s}_n)$ for $\mathcal{A}'$ which will be a PT-BNE, hence showing that $\mathcal{A}'$ is PT. 
To facilitate exposition, let us denote the interim utility of the $i$-th bidder under the truthful strategy $s^*$ in $\mathcal{A}$ by $g_i$, and note that it is given by 
$
g_i(v_i) = \mathbb{E}_{v_{-i} \sim F_{-i}} 
\left[  u_i(v_i, v_{-i})  \right].
$
With this in hand, $\bar{s}$ is specified as follows for the $i$-th bidder: 
\[
\bar{s}_i(e_i,v_i) =
\begin{cases} 
    v_i, & \text{if } g_i(v_i) \geq r(e_i) \\ 
    \abstain, & \text{otherwise}
\end{cases}
\]
To show that this is a PT-BNE of $\mathcal{A}'$, it suffices to consider the scenario where $\bar{s}(e_i,v_i) \neq \abstain$. In this scenario, notice that the following must hold for any $v_i,v_i' \in \mathbb{R}_{\geq 0}$ as $\mathcal{A}$ is truthful:
$$
g(v_i) = \mathbb{E}_{v_{-i} \sim F_{-i}} \left[  u_i(v_i, v_{-i})  \right] \geq \mathbb{E}_{v_{-i} \sim F_{-i}} \left[  u_i(v_i', v_{-i}) \right]
$$
If we subtract $r(e_i)$ from both sides, we get: 
$$
g(v_i) - r(e_i) \geq \mathbb{E}_{v_{-i} \sim F_{-i}} \left[  u_i(v_i', v_{-i}) \right] -r(e_i),
$$
where the LHS is the interim utility a player of type $(e_i,v_i)$ gets for reporting $v_i$, and the RHS is the interim utility for reporting $v_i'$. This shows that $\bar{s}$ is a PT-BNE. 
\end{proof}

Theorem \ref{thm:BNIC-to-PT} simplifies our search for welfare-optimal auctions with audits significantly. We have shown that it suffices to consider BNIC auctions in the classical Bayesian setting and simply control participation through the penalty function $r: \mathbb{R} \rightarrow \mathbb{R}$ as a means for the auctioneer to indirectly influence which bidders participate in the auction truthfully. With this in hand, for the rest of the paper we will only study auctions with audits that are based off of classical VCG auctions with the following allocation and payment rules:

\begin{definition}[VCG auction with Audits (VCG-A)]
We say that $\mathcal{A} = (x,p,r)$ is a VCG-auction with audits if it has the following allocation and payment rules for a joint bid profile $b = (b_1,\dots,b_n)$ such that $b_1 \geq b_2 \geq \dots \geq b_n$: 

\[
x_i(b) =
\begin{cases} 
    1, & \text{if } i \leq k \\ 
    0, & \text{otherwise}
\end{cases} \qquad\qquad
p_i(b) =
\begin{cases} 
    b_{k+1}, & \text{if } i \leq k \\ 
    0, & \text{otherwise}
\end{cases}
\]
\end{definition}
Note that our definition of VCG-A did not specify a penalty function $r$. This is due to Theorem \ref{thm:BNIC-to-PT}, which shows that any choice of $r$ results in a PT auction. Given our results from this section, we can refine the main objective of this paper: 

\begin{tcolorbox}[colframe=black,colback=gray!10,boxrule=0.5pt,arc=3pt]
\begin{center}
{\bf Our objective ($2^{nd}$ formulation):} \\
Find penalty functions, $r$, that maximize welfare in the VCG-A auction.
\end{center}
\end{tcolorbox}

\section{Participation Threshold Analysis}
\label{sec:participation-threshold-analysis}

A bidder $i$ participates ($b_i \neq \bot$) in a PT-BNE auction with audits $\mathcal{A}$ whenever their expected utility from participation is nonnegative:
$0 \leq \mathbb{E}_{v_{-i}\sim F_{-i}}[u_i(b)\mid b_i \neq \bot]$.
This condition induces a minimum valuation required for participation as a function of the bidder's externality. We refer to this boundary as the bidder's participation threshold and denote it by $\tau_i(e_i)$.
Because externalities affect bidder incentives only through participation, equilibrium behavior in PT auctions can be characterized entirely through these participation thresholds.
\begin{align*}
    0 &\leq
    \mathbb{E}_{v_{-i} \sim F_{-i}} 
\left[ x_i(b) \cdot  v_i  - p_i(b) - r(e_i)  \right] \\
    v_i &\geq \frac{\mathbb{E}_{v_{-i} \sim F_{-i}} 
\left[  p_i(v_i, v_{-i})  \right] + r(e_i)}{\mathbb{E}_{v_{-i} \sim F_{-i}} 
\left[  x_i(v_i, v_{-i})  \right]} \\
\tau_i(e_i) &=: \frac{\mathbb{E}_{v_{-i} \sim F_{-i}} 
\left[  p_i(\tau_i(e_i), v_{-i})  \right] + r(e_i)}{\mathbb{E}_{v_{-i} \sim F_{-i}} 
\left[  x_i(\tau_i(e_i), v_{-i})  \right]}
\end{align*}
If bidder valuations are independent and identically distributed (\textit{i.i.d.}) $v_i \stackrel{\mathrm{i.i.d.}}{\sim} F$, then for any two bidders $i,j$ with the same valuations $v_i = v_j$ the expected allocation probabilities and expected payments coincide for these bidders since the distribution of other bidder valuations is the same (i.e. $F_{-i}=F_{-j}$). Thus, we consider a common threshold function $\tau(e_i) = \tau_i(e_i) \;\; \forall {i\in[n]}$. This participation threshold behavior is visualized in Figure \ref{fig:participant_thresholds_illustrative}.

\begin{figure}[ht]
    \centering
    \includegraphics[width=0.5\linewidth]{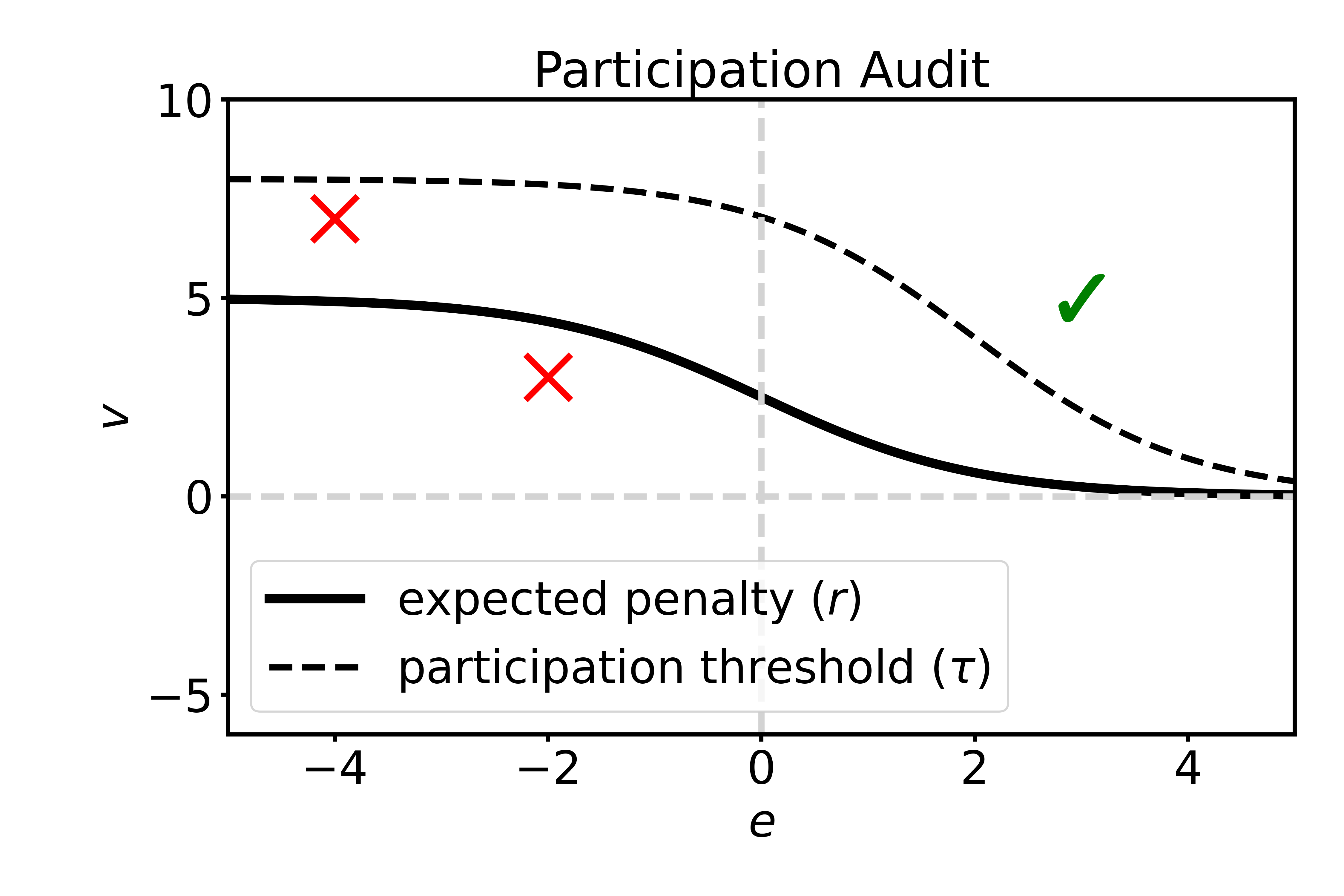}
    \caption{Participation thresholds in PT Auctions with Audits. Bidders participate whenever their value exceeds the threshold curve $\tau(e)$. The horizontal axis denotes bidder externality and the vertical axis denotes valuation. Green checks indicate participation while red crosses indicate non-participation.}
    \label{fig:participant_thresholds_illustrative}
\end{figure}

Notice, the penalty function $r$ may be fully defined by specifying an \textit{i.i.d.} marginal valuation distribution $F$ and participation threshold $\tau$.
$$r(e_i)= \mathbb{E}_{v_{-i} \sim F_{-i}} 
\left[  x_i(\tau(e_i), v_{-i})  \right]\tau(e) - \mathbb{E}_{v_{-i} \sim F_{-i}} 
\left[  p_i(\tau(e_i), v_{-i})  \right]$$
Thus, in the \textit{i.i.d.} setting we can search for optimal participation thresholds $\tau$ to find optimal penalties $r$. 

\begin{tcolorbox}[colframe=black,colback=gray!10,boxrule=0.5pt,arc=3pt]
\begin{center}
{\bf Our objective ($3^{rd}$ formulation):} \\
Find threshold functions, $\tau$, that maximize expected welfare in the VCG-A auction. 
\end{center}
\end{tcolorbox}

Characterizing optimal thresholds analytically requires reasoning about order statistics, since welfare depends on the contributions of the bidders with the largest realized values. Closed-form characterizations are generally difficult, but become tractable in the special cases studied in the following sections.

\subsection{No Competition}

When there is no competition, meaning there are at least as many goods as bidders ($k\geq n$), the optimal participation threshold is simply the zero-welfare curve,
\[
\tau^*(e)=-e.
\]
This threshold excludes exactly the bidders with negative welfare contribution while allowing all bidders with nonnegative welfare contribution to participate. Since all participating bidders can be allocated in the absence of competition, this allocation rule maximizes welfare.

\begin{proposition}
In the VCG auction with audits $\mathcal A$, when the number of goods satisfies $k\geq n$, the penalty function $r(e_i)=-e_i$ induces the participation threshold $\tau^*(e_i)=-e_i$ and maximizes welfare for any bidder types $\{(e_i,v_i)\}_{i=1}^n$.
\end{proposition}

\begin{proof}
When $k\geq n$, every participating bidder is allocated a good. Consequently, a bidder's participation does not affect any other bidder's allocation, so all VCG payments are zero. Since every participating bidder is allocated and all payments are zero, bidder $i$'s utility from participation is $u_i=v_i-r(e_i)$.
Thus bidder $i$ participates precisely when $v_i\geq r(e_i)$. Setting $r(e_i)=-e_i$ therefore induces the threshold $\tau_i(e_i)=-e_i$. Under this threshold, bidder $i$ participates if and only if $v_i+e_i=w_i\geq 0$.
Hence exactly the bidders with nonnegative welfare contribution participate and are allocated, which maximizes welfare.
\end{proof}

\subsection{Two Discrete Bidder Types}

Competition fundamentally changes welfare-optimal participation thresholds. In the no-competition setting, it is optimal to exclude exactly the bidders with negative welfare contribution. Under competition, however, even bidders with positive welfare contribution may optimally be excluded if their participation crowds out higher-welfare bidders. We characterize the conditions for which the welfare-optimal participation threshold excludes the minus type from participating.

\begin{proposition}
\label{prop:binary-threshold}
Consider allocating a single good (k=1) among $n$ bidders' with types drawn i.i.d.\ from a two-point mass distribution defining two types: plus type  $(e^+,v^+)$ and minus type  $(e^-,v^-)$. The plus type occurs with probability $h^+$ and the minus type occurs with probability $h^-=1-h^+$. Suppose the plus type has non-negative welfare contribution and strictly greater welfare contribution than the minus type.
\[
w^+ > w^-,\ w^+>0, 
\]
where
\[
w^+ = e^+ + v^+,
\qquad
w^- = e^- + v^-.
\]
In this setting, it is welfare-optimal for the audit penalty to exclude minus types from participating whenever $v^-\leq -e^-$, or when $v^- \geq v^+$ and 
\begin{equation} \label{eq:bin-threshold}
    w^-
\leq
w^+\left(
1-\frac{(1-h^+)^n}{1-(h^+)^n} \right).
\end{equation}
\end{proposition}

\begin{proof}
Let $M^+$ denote the number of plus-type bidders among the $n$ bidders. Then we know, 
\[
\mathbb P[M^+=0]=(1-h^+)^n,
\qquad
\mathbb P[M^+=n]=(h^+)^n.
\]
We divide all possible participation thresholds into five regimes:
$\mathcal C^0$ (no participation),
$\mathcal C^+$ (only plus types participate),
$\mathcal C^-$ (only minus types participate),
$\mathcal C^{+-}$ (both participate with $v^+ \geq v^-$),
and $\mathcal C^{-+}$ (both participate with $v^- \geq v^+$).

In regime $\mathcal C^{+}$, where only the plus type participates, expected welfare is
\[
\mathbb E[W\mid \mathcal C^+]
=
w^+\bigl(1-(1-h^+)^n\bigr).
\]
In regime $\mathcal C^{-+}$, where both types participate and $v^- \geq v^+$, a minus type is allocated whenever at least one minus type is present; otherwise, a plus type is allocated. Hence, 
\[
\mathbb E[W\mid \mathcal C^{-+}]
=
w^-\bigl(1-(h^+)^n\bigr)
+
w^+(h^+)^n.
\]
So, excluding the minus type is weakly better than allowing both types to participate whenever
\begin{align*}
\mathbb E[W\mid \mathcal C^+]
&\geq
\mathbb E[W\mid \mathcal C^{-+}]\\
w^+\bigl(1-(1-h^+)^n\bigr)
&\geq
w^-\bigl(1-(h^+)^n\bigr)+w^+(h^+)^n.
\end{align*}
Rearranging yields Equation \ref{eq:bin-threshold}.

Therefore, when $v^-\geq v^+$ and Equation \ref{eq:bin-threshold} holds, a threshold that excludes the minus type and admits the plus type maximizes expected welfare. This condition is depicted as the blue dashed threshold in Figure \ref{fig:binary-type-regions}.

In regime $\mathcal C^{+-}$, where both types participate and $v^+ \geq v^-$, a plus type is allocated whenever at least one plus type is present, while a minus type is allocated only when all bidders are minus type. Thus the expected welfare is 
\[
\mathbb E[W\mid \mathcal C^{+-}]
=
w^+\bigl(1-(1-h^+)^n\bigr)
+
w^-(1-h^+)^n.
\]
Comparing this to regime $\mathcal C^{+}$ where the the minus type is excluded, we see excluding the minus type when  $v^-\leq v^+$ results in greater welfare than allowing both types to participate only when $w^-> 0$. This is equivalent to the optimal no-competition participation threshold, $\tau(e)=-e$, and is depicted by the dashed red line in Figure \ref{fig:binary-type-regions}.

For completeness, consider the last two regimes. Regime $\mathcal C^{-}$, where only minus types participate, is always worse than when both types participate,  since $w^+>w^-$. Regime $\mathcal C^{0}$, where no types participate, is always worse than regime $\mathcal C^{+}$, where only the plus types participate, since $w^+>0$. If the welfare contribution of the plus type were also negative, $0\geq w^+$, then regime  $\mathcal C^{0}$ would be optimal, which is equivalent to the no-competition threshold.
\end{proof}

Figure~\ref{fig:binary-type-regions} depicts the optimal participation regimes of Proposition \ref{prop:binary-threshold}. If the minus type $(e^-,v^-)$ lies in the blue or red regions relative to the plus type $(e^+,v^+)$, then any penalty function which causes the minus type to not participate and allows the plus type to participate is welfare-optimal, ie. $v^- < \tau^*(e^-)$ and $v^+ > \tau^*(e^+)$.

\begin{figure}[ht]
    \centering
    \includegraphics[width=0.5\linewidth]{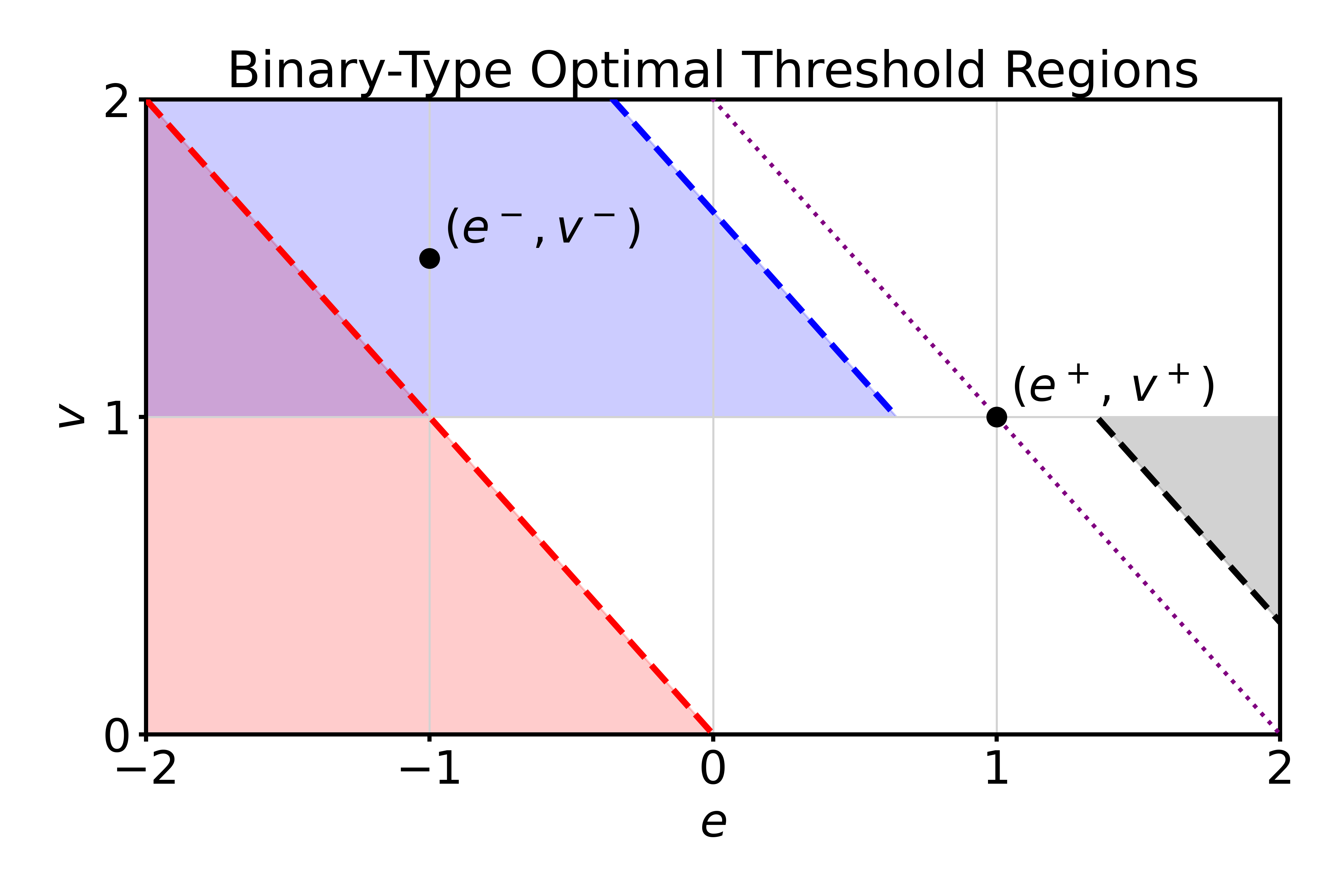}
    \caption{Regions of the type space, relative to the plus type $(e^+,v^+)$, in which all welfare-optimal thresholds exclude the minus type from participation. The dashed red line is $v=-e$. The dashed blue boundary is equal to Equation \ref{eq:bin-threshold} when $v^->v^+$. The black region is symmetric to the blue region, with the roles of the two types reversed.}
    \label{fig:binary-type-regions}
\end{figure}

This threshold can be interpreted as a welfare analogue of Myersonian reserve prices \cite{myerson_optimal_1981}: it may be optimal to leave goods unallocated in some realizations in exchange for higher expected welfare overall. Excluding minus types sacrifices welfare whenever no plus-type bidders are present, but increases welfare whenever a minus type would take allocation from a plus type with larger welfare contribution.

Equation~\ref{eq:bin-threshold} shows that minus types should participate only when their welfare contribution exceeds a sufficiently large fraction of the plus-type welfare contribution. This fraction depends on the likelihood that at least one competing plus type exists. Consequently, the optimal threshold becomes stricter as either the fraction of plus types ($h^+$) or the amount of competition ($n$) increases. Geometrically, the threshold boundary approaches the equal-welfare line $w^-\le w^+$ as competition grows.

\section{Recipient Audit Auctions}\label{sec:recipient-audit}

Previously, we have considered participant audits, where all participating bidders may incur penalties. Recipient audits are audits in which penalties are incurred only by bidders who ultimately receive a good. Whereas participant audits only affect participation, recipient audits directly modify bidders' effective valuations for receiving a good.

\begin{definition}[Auction with Recipient Audits]
An {\it Auction with Recipient Audits} is specified by a tuple:
$$
\mathcal{A}^R = (x,p,R)
$$
where each component is as follows:
\begin{itemize}
    \item $x: B^n \rightarrow \{0,1\}^n$ is the allocation function, where we use the shorthand $x_i(b)=x(b)_i \in \{0,1\}$ to represent whether bidder $i$ receives a good when joint bids $b \in B^n$ are made. We require that $x_i(\abstain,b_{-i})=0$ for any $b_{-i}\in B_{-i}$ and that $\sum_{i=1}^n x_i(b)\leq k$.

    \item $p: B^n \rightarrow \mathbb{R}^n$ is the payment function, where we use the shorthand $p_i(b)=p(b)_i \in \mathbb{R}$ to represent the amount paid by bidder $i$ under bid profile $b\in B^n$. We require that $p_i(\bot,b_{-i})=0$ for any $b_{-i}\in B_{-i}$.

    \item $R:\mathbb{R}\rightarrow\mathbb{R}$ is a recipient penalty function such that a bidder with externality $e_i$ incurs penalty $R(e_i)$ whenever they receive a good (i.e.\ whenever $x_i(b)=1$).
\end{itemize}
\end{definition}

Under recipient audits, bidder utilities become
$$
u_i(b) = (v_i - R(e_i)) \cdot x_i(b) - p_i(b).
$$

Thus, recipient penalties effectively shift bidders' valuations for receiving a good from $v_i$ to $\tilde v_i=v_i-R(e_i)$. Since we allow general penalty functions and general relationships between valuation and externality, recipient audits can arbitrarily change the distribution of `effective' valuations. In general, changing the valuation distribution may change BNE strategies. Unlike participant audits, this also breaks the separation between participation incentives and bidding incentives, and the revelation-principle-style result, previously established for participant audits.\footnote{Simulating best responses with recipient audits would require the auctioneer to know bidders' externalities directly.} Nevertheless, if the auction without audits has a dominant strategy equilibrium (DSE), then adding recipient audits induces a particularly simple transformation of bidder behavior, which we characterize in Theorem \ref{thm:recipient-audit-DSE}.

\begin{definition}[Dominant Strategy Equilibrium (DSE)] We say that $s^* = (s^*_1,\dots,s^*_n)$ is a {\it dominant strategy equilibrium} for auction $\mathcal{A}^R = (x,p,R)$ if, and only if, the following holds for each bidder, $i \in N$ with arbitrary  type $(e_i,v_i) \in \mathbb{R}^2$, each potential bid $b_i' \in B$, and all possible bid profiles of other bidders $b_{-i}\in B^{n-1}$ :
$$
 u_i(s^*_i(e_i,v_i), b_{-i}) \geq 
 u_i(b_i', b_{-i})
$$
\end{definition}

\begin{theorem}[Recipient Audit DSE]
\label{thm:recipient-audit-DSE}
Suppose that $s^*=(s_1^*,\dots,s_n^*)$ is a dominant strategy equilibrium of the auction without recipient audits $\mathcal{A}^0=(x^0,p^0)$. Then the strategy profile $s^R=(s_1^R,\dots,s_n^R)$ defined by
\[
s_i^R(e_i,v_i)=s_i^*(v_i-R(e_i))
\]
with the convention that $s_i^*(v_i-R(e_i))=\abstain$ whenever $v_i-R(e_i)<0$, is a dominant strategy equilibrium of the recipient-audit auction $\mathcal{A}^R=(x^0,p^0,R)$.
\end{theorem}

\begin{proof}
Since $s^*$ is a DSE of $\mathcal{A}^0$, we know for any set of bidders $i \in N$ with types $(v_i,e_i) \in \mathbb{R}^2$, each potential bid $b_i' \in B$, and all possible bid profiles of other bidders $b_{-i}\in B^{n-1}$:
\begin{align*}
 v_i \cdot x_i^0(s^*_i(e_i,v_i),b_{-i}) - p_i^0(s^*_i(e_i,v_i),b_{-i}) &\geq 
v_i \cdot x_i^0(b_i',b_{-i}) - p_i^0(b_i',b_{-i}) \\
\end{align*}
Since this is assumed to be true for any possible valuation $v_i\in\mathbb{R}$, it must also be true for valuation $\tilde{v}_i = v_i - R(e_i)$ since it is a real number $\tilde{v}_i\in\mathbb{R}$. 
\begin{align*}
 \tilde{v}_i \cdot x_i^0(s^*_i(e_i,\tilde{v}_i),b_{-i}) - p_i^0(s^*_i(e_i,\tilde{v}_i),b_{-i}) &\geq 
\tilde{v}_i \cdot x_i^0(b_i',b_{-i}) - p_i^0(b_i',b_{-i}) \\
\end{align*}
This is exactly the necessary and sufficient requirement for $s^*_i(e_i,\tilde{v}_i)$ to be a DSE of $\mathcal{A}^R$.
\end{proof}

To understand the implications of Theorem~\ref{thm:recipient-audit-DSE}, consider recipient audits added to an auction with truthful DSE, such as VCG. In this case, equilibrium bids become $b_i=v_i-R(e_i)$. Bidders with $v_i<R(e_i)$ have negative effective value and rationally abstain, so the participation threshold for recipient audits is simply equal to the penalty function, $\tau(e_i)=R(e_i)$.

For any bid amount $b'$, the points $(e_i,v_i)$ satisfying $b'=v_i-R(e_i)$ form an iso-bid line in type space. Geometrically, these iso-bid lines are vertical shifts of the penalty curve $v=R(e_i)$. Since any auction with a truthful DSE (which we call a DSIC Auction: Dominant Strategy Incentive Compatible) must allocate to the $k$ highest bids (as in VCG), then the iso-bid lines are also the iso-allocation lines. Figure~\ref{fig:recipient-iso-lines} compares these iso-bid lines against the iso-welfare-contribution lines, where welfare contribution is given by $w_i=e_i+v_i$. This geometric relationship immediately yields Theorem~\ref{thm:welfare-maximal-recipient-audit}.

\begin{figure}[ht]
    \centering
    \includegraphics[width=0.9\linewidth]{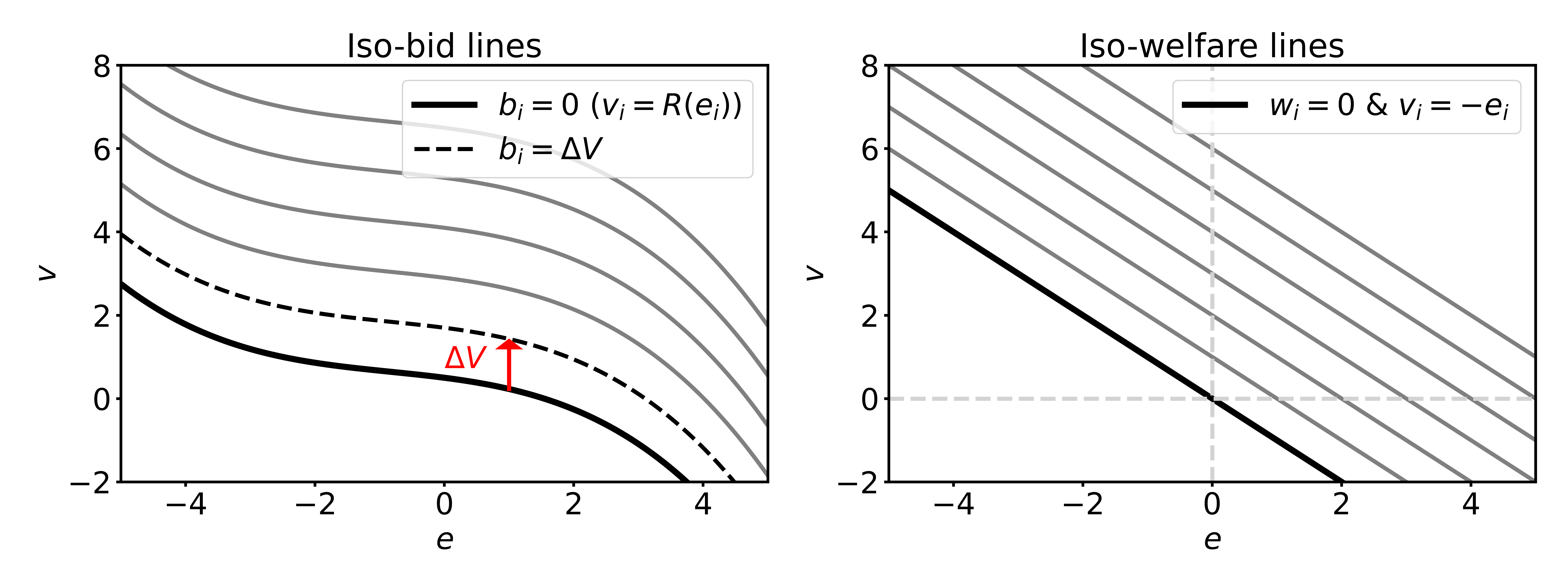}
    \caption{Visualization of the level set lines of bids and welfare contributions in auctions with recipient audits.}
    \label{fig:recipient-iso-lines}
\end{figure}

\begin{theorem}[Welfare Maximization with Recipient-Audited VCG]
\label{thm:welfare-maximal-recipient-audit}
For any bidders $i\in[n]$ with types given by $(e_i,v_i)_{i\in[n]}$, the recipient-audit auction
\[
\mathcal{A}^*=(x^{vcg},p^{vcg},R),
\]
where $(x^{vcg},p^{vcg})$ are the VCG allocation and payment rules and
\[
R(e_i)=-e_i,
\]
results in the welfare-maximizing allocation in DSE.
\end{theorem}

\begin{proof}
Since VCG is DSIC (i.e. has a truthful DSE), Theorem~\ref{thm:recipient-audit-DSE} implies that DSE bids in $\mathcal{A}^*$ are
\[
b_i=v_i-R(e_i)=v_i+e_i=w_i,
\]
meaning bids are exactly equal to bidders' welfare contributions. Since the VCG allocation rule allocates to the $k$ highest bids, in DSE $\mathcal{A}^*$ allocates to the $k$ bidders with greatest welfare contribution. This maximizes welfare over all possible allocations.
\end{proof}

\subsection{Comparing Participant and Recipient Audits}

Recipient audits alter equilibrium bidding behavior itself rather than only precluding certain bidders from participating. In some applications recipient audits may be more realistic than participant audits if receiving a good is a more externally visible outcome. While in other applications participant audits may be easier to implement.

Recipient audits admit a particularly strong welfare-maximization result: VCG allocation and payment rules combined with the recipient penalty $R(e_i)=-e_i$ causes equilibrium bids to coincide exactly with individual welfare contributions, maximizing welfare including externalities. However, implementing this mechanism requires the auctioneer to impose expected penalties (and rewards) exactly equal to bidders' externalities. 

In many applications, exactly implementing an optimal penalty function may not be practically feasible. This motivates studying constrained variants of participant and recipient audits. We hypothesize there may be forms of penalty constraints such that participant audits achieve greater expected welfare despite their weaker guarantees when unconstrained. More broadly, participant and recipient audits represent two distinct approaches to incorporating externalities into auction design: participant audits control welfare through endogenous participation, while recipient audits directly modify effective bidder valuations.

Additionally, our theory for recipient and participant audits compose nicely, allowing them to be combined. Studying constrained combinations of participant and recipient audits appears to be a rich direction for future work. Analytical characterizations are likely to require reasoning about order statistics, suggesting computational approaches may ultimately prove more fruitful.

\section{Limitations and Dangers of Audits}\label{sec:dangers_of_audits}

Even though optimal unconstrained audits are welfare-improving in theory, constrained or imperfect audits may become welfare-perverse in practice. In particular, beneficial bidders with lower willingness to pay may be excluded before harmful bidders if audits are insufficiently accurate. 

To illustrate this phenomenon, consider two bidder types: a plus type with positive externality and a minus type with negative externality. The plus type may represent, for example, a smaller entrant firm advertising a high-quality product to a specialized market, or a speaker providing socially beneficial information to a minority community. The minus type may instead represent a large incumbent with stronger private incentives but negative social externalities. For concreteness, consider the numerical example
\[
(e^+,v^+)=(1,1),
\qquad
(e^-,v^-)=(-3,2).
\]
The minus type has larger willingness to pay, but negative net welfare contribution. Thus, in the unconstrained setting, it is always welfare-optimal to preclude the minus type from participating via participant and/or recipient audits.

Suppose the audit mechanism is constrained as follows. Each individual is audited independently with probability $q=\frac{1}{10}$, and audited bidders incur a penalty with probability $\ell^+=\frac{4}{10}$ or $\ell^-=\frac{5}{10}$ depending on whether they are plus or minus types, respectively. We interpret $\ell^+$ and $\ell^-$ as audit accuracies (e.g.\ true positive and true negative rates). The remaining design parameter is the penalty magnitude $\bar r$, yielding expected recipient penalties
\[
R(e^+)=(q\ell^+) \bar r,
\qquad
R(e^-)=(q\ell^-) \bar r.
\]

For simplicity, we consider recipient audits. A bidder of type $(e',v')$ is excluded whenever their penalty-adjusted bid becomes negative: $v' < q\ell' \bar r$.
Thus, excluding the harmful minus type requires $\bar r > \frac{v^-}{q\ell^-}$. However, sufficiently large penalties may also exclude the beneficial plus type. In particular, combining the exclusion conditions yields the following simple comparison of the ratio of valuations to the ratio of audit accuracies: 
\[
\frac{v^-}{v^+} > \frac{\ell^-}{\ell^+}.
\]
When this condition holds, it becomes impossible in this example to exclude the harmful minus type without also excluding the beneficial plus type. In our numerical example,
\[
\frac{v^-}{v^+}=2
\qquad\text{and}\qquad
\frac{\ell^-}{\ell^+}=\frac{5}{4},
\]
so the beneficial type is excluded first. Figure~\ref{fig:binary-constrained-audits} visualizes this.

\begin{figure}[ht]
    \centering
    \includegraphics[width=0.5\linewidth]{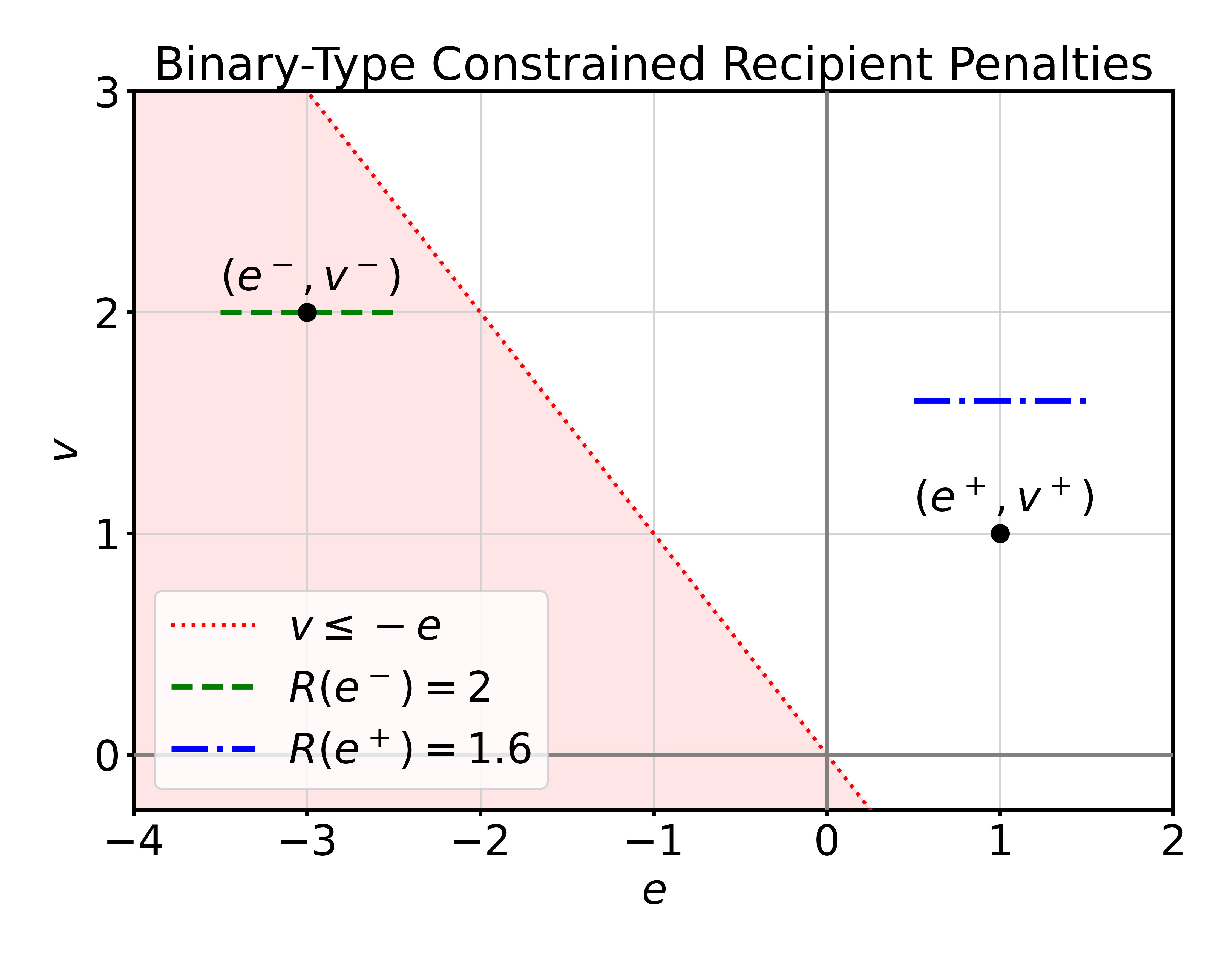}
    \caption{Constrained recipient audits can become welfare-perverse under imperfect audit accuracy. In this example, the beneficial plus type $(e^+,v^+)$ is excluded before the harmful minus type $(e^-,v^-)$ because the audit system insufficiently distinguishes between the two types.}
    \label{fig:binary-constrained-audits}
\end{figure}

This example illustrates a broader issue with constrained audits: excluding harmful bidders may also exclude beneficial bidders with lower willingness to pay. As audit accuracy decreases, this effect becomes more severe. In the extreme, audits could apply greater penalties to more beneficial types, becoming welfare-perverse by excluding socially beneficial bidders before harmful, or less beneficial, bidders.

Natural processes may behave like perverse or beneficial audits penalties. For example, speakers sharing information inconvenient to powerful institutions may face greater expected retaliation, functioning like a perverse recipient audit. Likewise, producing carefully researched information may require substantially greater effort than producing misleading or fabricated content, functioning like a perverse participant audit.

The numerical example in this section also shows that excessively large penalties may be required given low audit probabilities and imperfect audit accuracy. In our example the required penalty is twenty times larger than the harmful bidder's value, $\frac{\bar r}{v^-}=\frac{40}{2}=20$. The penalty magnitude required to exclude harmful bidders increases as either audit accuracy, and/or audit probability, decreases. Extremely large penalties may be infeasible in practice, and may disproportionately affect risk-averse or budget-constrained participants, creating additional fairness concerns.

Moreover, the chances of audits occurring may be constrained by the cost of performing audits in many applications. Consequently, practical audit systems may face a fundamental tradeoff between audit cost, audit accuracy, and welfare performance. We do not consider audit cost in our main theoretical results because we studied unconstrained audits. In principle, unconstrained audits with rational bidders could use infinitesimal audit probabilities together with arbitrarily large penalties, making expected audit cost negligible.

\section{Numerical Solutions and Empirical Welfare Impacts}
\label{sec:numerical-empirical}

\subsection{Methods}

To evaluate audit-based mechanisms empirically, we simulate auctions with $n$ agents competing for $k$ allocations. Each agent has a valuation $v_i$ and externality $e_i$. We compare three auction formats: a standard VCG auction ($vcg$), a participant-audited VCG auction ($vcgPA$) in which participation is governed by a threshold function $\tau(e)$, and a recipient-audit auction ($vcgRA$) in which allocation occurs in decreasing order of social welfare contributions, $w_i=v_i+e_i$. This is consistent with our preceding theory.

\paragraph{Auction Mechanisms in Simulation.}
For each simulated auction draw $j$ -- formed by sampling $n$ agents with $(v_i,e_i)$ pairs from the relevant population -- we evaluate all three mechanisms on the same draw. The \emph{VCG} mechanism allocates to the $k$ highest-value agents,
\[
S_j^{\text{vcg}} = \argmax_{S\subseteq[n],\,|S|=k} \sum_{i\in S} v_i,
\qquad
W_j^{\text{vcg}} = \sum_{i\in S_j^{\text{vcg}}} (v_i+e_i).
\]
The \emph{Participant Audit} mechanism ($vcgPA$) admits only agents whose declared value clears the threshold,
$\mathcal{M}_j(\bsc) = \{i\in[n] : v_i \geq \tau(e_i;\bsc)\}$,
where $\bsc$ denotes the vector of threshold-function parameters being optimized (made explicit below),
and allocates to the top $k$ admitted agents by value,
\[
S_j^{\text{vcgPA}}(\bsc) = \argmax_{S\subseteq \mathcal{M}_j(\bsc),\,|S|\le \min(k,\,|\mathcal{M}_j(\bsc)|)} \sum_{i\in S} v_i,
\qquad
W_j^{\text{vcgPA}}(\bsc) = \sum_{i\in S_j^{\text{vcgPA}}(\bsc)} (v_i+e_i),
\]
with payment to admitted agents set by a generalized $(k{+}1)$th-price rule: the $(k{+}1)$th-highest value among admitted agents if at least $k+1$ are admitted, and zero otherwise. The \emph{Recipient Audit} mechanism ($vcgRA$) applies no participation threshold and instead allocates to the $k$ agents with the largest positive social value $w_i=v_i+e_i$,
\[
S_j^{\text{vcgRA}} = \argmax_{S\subseteq\{i\,:\,w_i>0\},\,|S|\le k} \sum_{i\in S} w_i,
\qquad
W_j^{\text{vcgRA}} = \sum_{i\in S_j^{\text{vcgRA}}} w_i.
\]
For each experimental configuration we draw $J_{\text{train}}=500$ training auctions, used exclusively as the genetic algorithm's fitness signal during threshold optimization, and a disjoint set of $J_{\text{eval}}=2000$ evaluation auctions, held out from the optimizer and used only for final welfare reporting; all welfare figures reported below average $W_j^{\text{vcg}}$, $W_j^{\text{vcgPA}}(\bsc^*)$, or $W_j^{\text{vcgRA}}$ over the evaluation draws $j\in[J_{\text{eval}}]$ .

For participant-audit auctions, we numerically optimize threshold functions $\tau(e)$ to maximize expected welfare using a genetic algorithm. We parameterize $\tau(e)$ as a low-degree polynomial, $\tau_d(e)=\sum_{j=0}^d A_j e^j$, with coefficient vector $\bsc=(A_0,\dots,A_d)$ and degree $d$ swept over $\{1,2,3\}$; for each candidate $\bsc$, welfare is estimated by simulating auctions in which agents satisfying $v_i\geq\tau(e_i;\bsc)$ participate and the top-$k$ participating bidders by valuation are allocated. Full details of the genetic algorithm (selection, crossover, mutation, restart strategy, and hyperparameters) are provided in \Cref{appendix:ga-details}; the optimization code itself is available in the released supplemental materials (Appendix~\ref{appendix:git-repo}).

\subsection{Simulated Datasets}
\label{appendix:simulated-data-text}

We first evaluate the proposed mechanisms on synthetic datasets drawn from several joint distributions over valuations and externalities. In addition to the representative unimodal Gaussian setting described below, we also generate bimodal lateral, bimodal positively correlated, bimodal negatively correlated, and quadmodal joint distributions over $(e,v)$. For each distribution we draw a population of 10,000 agents and sample from that population to optimize polynomial participation thresholds $\tau(e)$ of degrees  $d\in\{1,2,3\}$ for each combination of bidder count $n\in\{8,16\}$ and allocation count $k\in\{1,2,4\}$. Robustness figures for these additional distributions are provided in Appendix~\ref{appendix:robustness-figures}.

For the representative unimodal setting, agents are sampled from
\[
e \sim N(0,0.1),
\qquad
v \sim N(2,0.1),
\]
with negative valuations removed.

Figure~\ref{fig:penalty_by_degree_sim} illustrates representative optimized threshold functions and their corresponding participant penalties, and optimal recipient penalties. The induced penalty curves are computed from the optimized thresholds as described in \cref{appendix:penalty-curves}. Consistent with the theoretical results of \cref{sec:participation-threshold-analysis,sec:recipient-audit}, the optimized thresholds selectively exclude low-welfare bidders while preserving most high-valuation allocations, resulting in net welfare increases.\footnote{The flattening of the participant penalty curve near zero is an artifact due to the nature of our $(k+1)$th price rule. There is zero chance of winning an item with valuations near and below zero. Thus the necessary penalty to dissuade participation approaches zero. The numerically optimized threshold and penalty curves are not meaningful in regions far away from the samples of the bidder population because there is no training signal in these regions.}

\begin{figure}[ht]
    \centering
    \includegraphics[width=\linewidth]{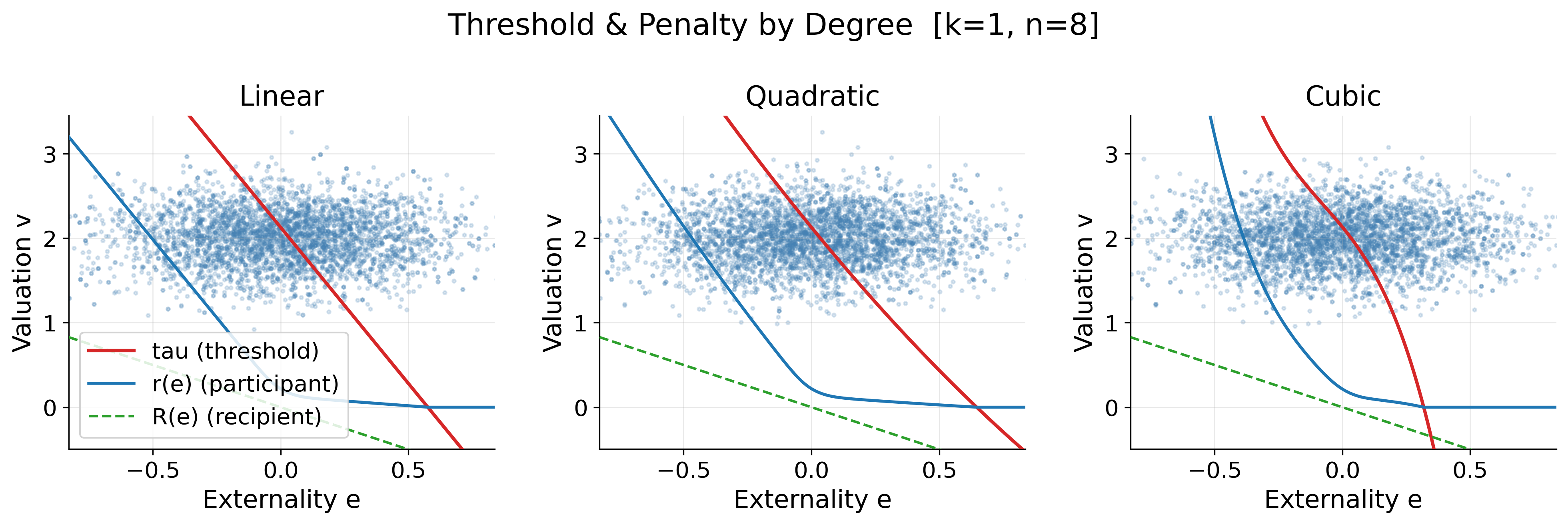}
    \caption{Representative optimized participant threshold functions together with their induced participant, and optimal recipient penalties, for the unimodal synthetic dataset ($k=1$, $n=8$). The optimized thresholds selectively exclude low-welfare bidders while preserving most high-valuation allocations. }
    \label{fig:penalty_by_degree_sim}
\end{figure}

Figure~\ref{fig:welfare_dist_sim} compares net welfare, valuation, and externality of allocated bidders across auction mechanisms. Both participant-audit and recipient-audit auctions shift welfare distributions substantially to the right while only modestly reducing bidder valuations. Recipient audits produce the largest welfare gains, consistent with the welfare-maximization results established theoretically in \ref{sec:recipient-audit}.

\begin{figure}[ht]
    \centering
    \includegraphics[width=\linewidth]{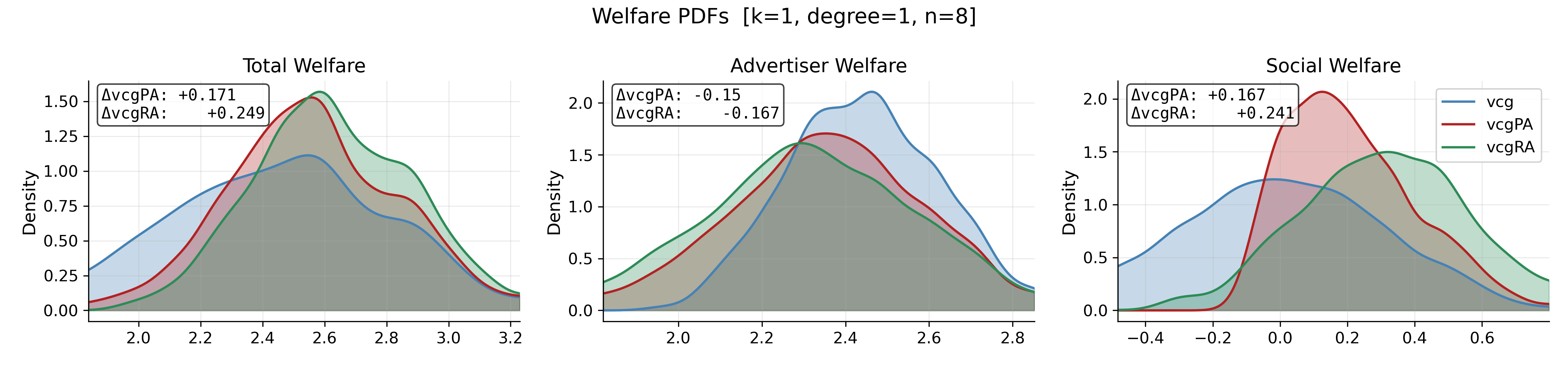}
    \caption{Distributions of welfare, valuations, and externalities under standard VCG, participant-audit, and recipient-audit mechanisms on the unimodal synthetic dataset. Audit-based mechanisms substantially improve welfare and externality outcomes while only modestly reducing valuations. $\Delta$vcgPA and $\Delta$vcgRA values are the integral differences between the VCG PDF and the participant audited and recipient audited PDFs respectively.}
    \label{fig:welfare_dist_sim}
\end{figure}

\subsection{Empirical Dataset Construction}
\label{sec:data-collection}

Modeling audited auctions in real applications requires an empirical joint distribution over bidder valuations and externalities. Because existing works with  rigorous joint estimates of valuation and externality have proprietary datasets \cite{shapiro_positive_2018, sinkinson_ask_2015, dubois_effects_2016, tuchman_advertising_2019}, we contribute a small open dataset and a rough method for estimating valuation and externality, consisting of approximately 11,800 X-Twitter posts annotated with Community Notes (which we refer to as the XNP400 dataset). Our method for estimating valuation and externality should only be considered to provide relative rankings of valuation and externality among posts, or perhaps a rough estimate of the shape of the joint distribution.

We construct the dataset using publicly available Community Notes data together with post-level metadata retrieved from the X API, as of February~2,~2025. Community Notes are labels which any User on X may choose to apply to a post, regarding how misleading the post is; other users may also rate the helpfulness of any Note. We filter the dataset to remove extreme data points. After filtering, the resulting dataset contains $N=11{,}853$ posts; a full description of data collection and filtering is given in \Cref{appendix:data-collection}.

For each post $i$, we estimate an externality $e_i = \zeta\cdot\beta_i/\mu_i$, where $\mu_i$ is the post's age in months, $\zeta$ is an externality-magnitude scaling parameter, and $\beta_i$ is an aggregate post rating constructed from Community Notes labels and helpfulness ratings: Notes labeled \texttt{ MISINFORMED\_OR\allowbreak\_POTENTIALLY\_MISLEADING} contribute negatively to $\beta_i$ and Notes labeled \texttt{NOT\_MISLEADING} contribute positively, with each Note weighted by the volume of ratings in support and opposition of the Note. The full scoring procedure is given in \Cref{appendix:externality-estimation}. We estimate valuations using post engagement as a proxy for advertiser value, $v_i = \theta\cdot\alpha_i/\mu_i$, where $\alpha_i$ is the total number of actions (likes, replies, reposts, clicks) on post $i$; we fix valuation magnitude $\theta=1$ throughout and instead vary $\zeta$, so that valuations and externalities are interpreted in units relative to each other. We believe $\theta=1$ is a realistic average magnitude. Additional discussion of the valuation proxy is given in \Cref{appendix:valuation-estimation}. These methods effectively assume linear relationships between ratings and externality, and between engagement and valuation. 

Because this model of externalities is crude, and the true magnitudes of externalities and valuations are uncertain, we simulate audited auctions across a range of $\zeta$ values. Descriptive statistics and correlation structure for the constructed dataset are reported in \Cref{appendix:descriptive-statistics}.

\FloatBarrier

\subsection{Empirical Welfare Results}

We optimize participant-audit thresholds across varying numbers of items allocated $k$, numbers of bidders $n$, and externality scales $\zeta$. 

Some prior studies provide empirical dollar-value estimates of average externality magnitudes from online advertising. Goldstein et al.~\cite{goldstein_economic_2014} estimate a CPM (cost per thousand impressions) of \$1.53 for spam impressions, while Schnadower Mustri et al.~\cite{schnadower_mustri_behavioral_2023} estimate a cost of \$0.54 per interaction with a targeted advertisement. We convert each into an externality-cost calibration $\zeta$ by dividing by the sample mean of aggregate rating per impression or per action, respectively,
\[
\zeta^{\text{imp}} = \frac{1.53/1000}{\overline{\beta_i/\eta_i}} \approx \$2.71,
\qquad\qquad
\zeta^{\text{act}} = \frac{0.54}{\overline{\beta_i/\alpha_i}} \approx \$2.54,
\]
in dollars per rating unit per month, where $\eta_i$ denotes post $i$'s lifetime impressions (see \cref{appendix:zeta-calibration} for details). The two estimates agree closely, yielding a baseline calibration of $\zeta\approx\$2.5$. Because this magnitude is approximate, we consider a range of externality magnitudes; our representative grid value $\zeta=2.8$ used in the sweep below sits close to both calibration estimates. 

We search for optimal participant-audit thresholds among polynomials of degrees $1$, $2$, and $3$, for all combinations of parameters in the following sets.
\[
\zeta\in\{1,2.8,4.6,6.4,8.2,10\},
\qquad
k\in\{1,2,4,8\},
\qquad
n\in\{8,16\}.
\]

Figure~\ref{fig:tau_degree_comparison_emp_appendix} compares optimized threshold functions for varying values of $\zeta$ in the representative setting $k=1,n=8$.

\begin{figure}[ht]
    \centering
    \includegraphics[width=0.8\linewidth]{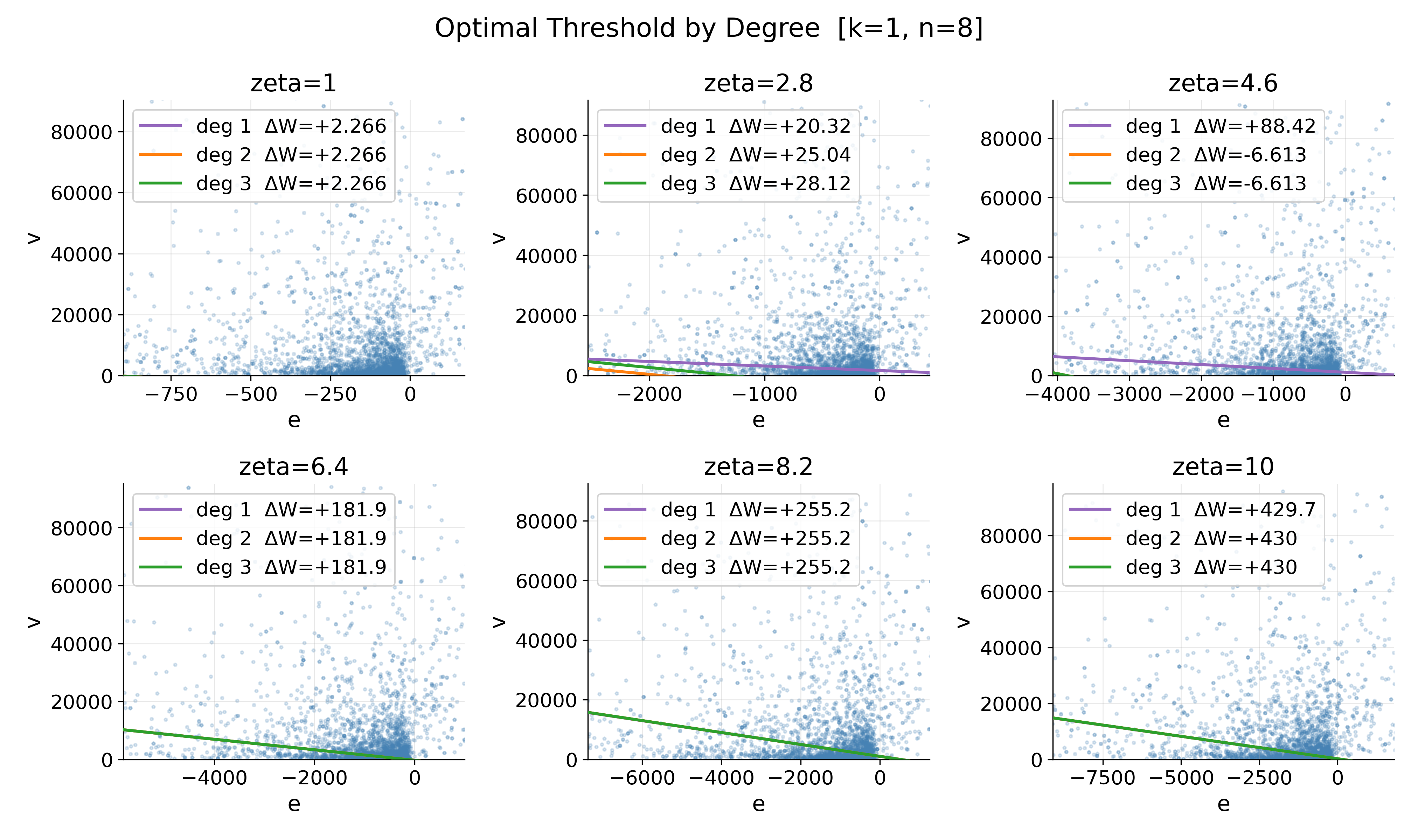}
    \caption{Optimized participant threshold functions for varying threshold polynomial degrees and externality scales $\zeta$. For small values of $\zeta$, optimized thresholds primarily exclude extreme negative-externality outliers while leaving most of the distribution unaffected. As $\zeta$ increases, the optimized thresholds become substantially more selective. The genetic optimization exhibits more instability and sub-optimality for smaller values of $\zeta$.}
    \label{fig:tau_degree_comparison_emp_appendix}
\end{figure}

\subsubsection{Estimation of the Scale of Social Impact}
\label{sec:welfare-comparisons}

Figure~\ref{fig:welfare_by_zeta_emp} compares expected welfare under standard VCG auctions, best-found participant-audit auctions, and optimal recipient-audit auctions. Across nearly all tested parameter settings, both audit-based mechanisms improve welfare relative to standard VCG allocation, with recipient audits producing the largest gains as expected. As externality magnitude $\zeta$ increases, the relative welfare gains from both audit mechanisms become substantially larger, with an apparently non-linear trend 

\begin{figure}[ht]
    \centering
    \includegraphics[width=\linewidth]{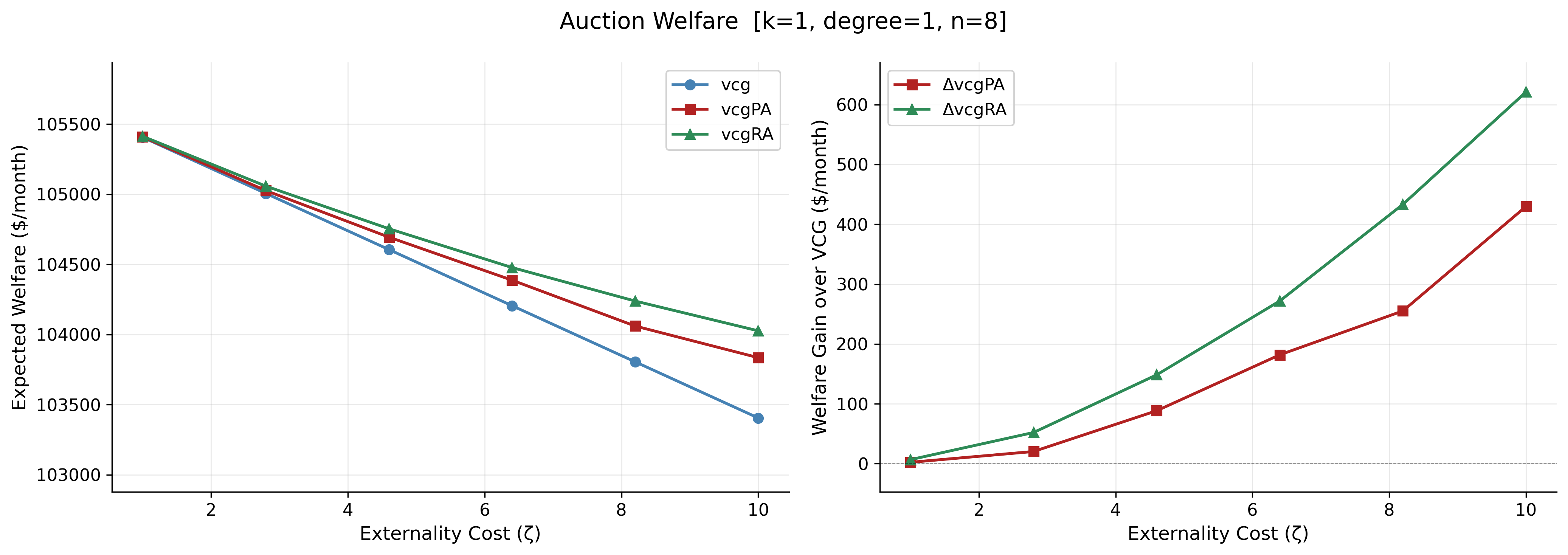}
    \caption{Expected welfare from standard VCG, participant-audit, and recipient-audit mechanisms for varying externality scales $\zeta$ on the X Post dataset. Expected welfare is calculated as the average valuation + externality of the winners of the evaluation set of simulated auctions. The welfare values represent the generated welfare per month by an ad that has won an auction. Absolute welfares decrease with $\zeta$ because negative externalities receive larger weight in welfare calculations, reducing baseline welfare from standard VCG allocation. Participant audit welfare gain (i.e. the difference between participant audited auction welfare and VCG auction welfare) at externality cost per rating $\zeta=2.8$ (0.59 \$/action): \$20; at externality cost per action = $\zeta=10$ (2.13 \$/action): \$430. Recipient audit welfare gain at the same $\zeta$ values are \$52 and \$621 respectively.}
    \label{fig:welfare_by_zeta_emp}
\end{figure}

As a rough illustrative exercise, we combine these welfare estimates with estimates of advertising activity on large online advertising platforms. A thorough previous study found there were at least 267{,}000 political advertisements published on Facebook over three months in 2018 \cite{laura_edelson_nyus_nodate, new_york_university_tandon_school_of_engineering_now_2018, edelson_analysis_2019, edelson_security_2020}---roughly $89{,}000$ ads per month. Naively scaling our simulated per-auction welfare improvements (see caption of Figure \ref{fig:welfare_by_zeta_emp}) by this volume gives total welfare gains across a single platform, as shown in table \ref{tab:welfare_gain_estimates}.

\begin{table}[htbp]
\centering
\caption{Estimated Platform-total Welfare Gains Per Month}
\label{tab:welfare_gain_estimates}
\small
\begin{tabular}{lll}
\toprule
&   $\zeta\approx2.8$&$\zeta\approx10$\\
\midrule
$\Delta W^{\text{vcgPA}}$& $89{,}000\times\$20\approx\$1.8\text{M per month}$&$89{,}000\times\$430\approx\$38\text{M per month}$\\
$\Delta W^{\text{vcgRA}}$& $89{,}000\times\$52\approx\$4.6\text{M per month}$&$89{,}000\times\$621\approx\$55\text{M per month}$\\
\bottomrule
\end{tabular}
\end{table}

We use $\zeta\approx 2.8$ as the baseline because it closely matches the two empirical estimations of user externality mentioned above. However, the externality magnitude $\zeta\approx 10$ could be interpreted as assuming a platform user's externality has the baseline magnitude of $\zeta\approx2.8$ while the social externality (broader members of society) is approximately twice as large as the user's, resulting in a sum total approximately triple the baseline. Thus, even relatively small improvements at the auction level could aggregate into substantial system-level welfare benefit, which we estimate to be on the order of \$1--10 million per month just from political advertisement on a single major social media platform in 2018. However, this estimate should be interpreted only as an illustrative order of magnitude, rather than a quantitative forecast. 

More broadly, the empirical results should be interpreted primarily as evidence that audit-based mechanisms can meaningfully alter welfare outcomes under plausible externality structures, rather than as precise estimates of platform-wide welfare effects. The constructed empirical dataset is not fully representative, and likely differs from true distributions of promoted advertisements on media platforms. The primary purpose of the empirical analysis is therefore to illustrate the qualitative behavior of audit-based mechanisms and to demonstrate that welfare improvements remain plausible under noisy and imperfect real-world proxy signals.

\subsection{Discussion of Empirical Results}

The empirical results suggest that audit-based auction mechanisms can substantially improve welfare relative to standard VCG allocation under a range of externality scales and auction configurations. 

Across most parameter combinations tested, affine threshold functions perform comparably to higher-order polynomial thresholds. In several cases, quadratic and cubic thresholds converge closely to affine solutions, suggesting that simple linear participation rules may capture much of the achievable welfare improvement. This also suggests analytical approximation results may be tractable despite the order statistics involved.

Several important limitations remain. First, the empirical distribution depends heavily on the externality scaling parameter $\zeta$, which is inherently difficult to estimate and likely varies substantially across types of content and contexts. Second, the dataset itself provides only a rough proxy for externalities and valuations. Community Notes primarily capture misinformation-related harms and do not measure many other forms of externality. Moreover, because promoted posts cannot be reliably identified through the X API, many sampled posts likely do not correspond to actual advertisements.

The empirical optimization procedure also has important limitations. Higher-order threshold functions occasionally exhibited unstable out-of-sample performance, possibly due to overfitting to the discrete auction samples during optimization. 

\section{Conclusion}\label{sec:conclusion}

This paper introduced auctions with audits, a mechanism-design framework in which bidders face externality-dependent penalties that induce endogenous participation constraints, and/or change effective bidder valuations. We showed that these mechanisms can be implemented using standard truthful auctions, such as VCG, augmented with personalized and non-manipulable participation penalties. This induces participation thresholds as a function of bidder externalities and allows auction participation itself to become a mechanism for screening socially harmful allocations. Our results can also be used to understand the behavioral effects of other natural, or non-natural, exogenous processes which function like audits and penalties. 

We characterized welfare-optimal participation thresholds in several analytically tractable settings. In particular, we showed that affine thresholds are welfare-optimal in the absence of competition and in binary-type settings with a single good. More broadly, our numerical experiments suggest that affine thresholds and/or penalties may remain close to optimal in substantially richer environments as well. 

Our primary motivation comes from large-scale attention markets such as digital advertising and social media promotion systems, where substantial externalities remain unpriced despite growing societal concern regarding misinformation, manipulation, and harmful content amplification. Audit-based mechanisms provide an alternative to centralized welfare estimation by allowing externality-sensitive participation constraints to emerge through decentralized auditing and penalty systems. More generally, we believe related mechanisms may have applications in domains involving information elicitation, reputation systems, tax auditing, blockchain oracles, journalism, online reviews, and government contracting.

To explore the empirical implications of audit-based auctions, we constructed a dataset of approximately twelve thousand X-Twitter posts associated with Community Notes and used this dataset to estimate approximate joint distributions over advertiser valuations and externalities. Simulations suggest that audit-based mechanisms can substantially improve welfare relative to standard VCG allocation while only modestly reducing bidder valuations. In simulations recipient-audit mechanisms achieved the largest welfare gains by construction, while participant-audit mechanisms produced substantial improvements using relatively simple affine threshold structures.

Several important directions for future work remain. On the theoretical side, an important open problem is characterizing welfare-optimal thresholds and penalties under general type distributions and richer competitive environments. Another natural extension is to study settings with explicit limits on auditing capacity, imperfect or noisy audits, or restrictions on admissible penalty functions. There is also opportunity to extend the theory to more general auction and mechanism design settings, such as position auctions. 

On the empirical side, more accurate measurements of the joint distribution of information externalities and advertiser valuations remain an important challenge. Improved datasets, richer models of user welfare, and endogenous behavioral responses by platforms and advertisers could substantially improve future empirical evaluations. Different numerical optimization methods could be explored to improve stability or attain optimality guarantees.

Overall, our results suggest that audit-based mechanisms provide a promising approach for incorporating externalities into auction design, particularly in large-scale information and attention markets where centralized welfare estimation may be infeasible or undesirable.

\section*{Acknowledgements}

The authors declare no conflict of interest and have received no financial support for this work. A peer reviewed National Science Foundation grant SaTC \#2217770 was awarded to Marshall Van Alstyne, but it was rescinded without explanation by the Department of Government Efficiency (DOGE) before it could support this work.  

\bibliographystyle{ACM-Reference-Format}
\bibliography{references}

\appendix \label{appendix_start}

\section{Optimization and Simulation Details}
\label{appendix:optimization}

This appendix gives the full details of the auction simulation, the genetic algorithm used to optimize participation thresholds, and the expected penalty curves induced by the optimized thresholds. The three auction mechanisms ($vcg$, $vcgPA$, $vcgRA$) and the welfare quantities they produce are defined in the main text (\cref{sec:numerical-empirical}); here we focus on the components specific to the optimization procedure.

\subsection{Auction Draws and Fair Comparison}
\label{appendix:auction-draws}

For each experimental cell $(k,\zeta,n,D)$ -- allocation slots, externality scale, bidders per draw, and threshold polynomial degree -- we maintain two independent sets of auction draws with distinct purposes. The \emph{training draws} ($J_{\text{train}}=500$) are used exclusively by the genetic algorithm as its fitness signal during optimization. The \emph{evaluation draws} ($J_{\text{eval}}=2000$) are held out from the optimizer entirely and used only for final welfare reporting; the larger draw count reduces Monte Carlo variance in the reported welfare estimates. Each draw $j$ is formed by sampling $n$ agents without replacement, where agent $i$ enters with its pair $(v_i,e_i)$ (drawn from the relevant synthetic distribution or from the empirical dataset).

To ensure that observed welfare differences reflect the parameters of interest rather than sampling noise, the random seeds controlling the draws are made functions of $(\zeta,n)$ only, excluding $k$ and $D$. Consequently, all degrees and slot counts within the same $(\zeta,n)$ cell are trained and evaluated on identical agent samples, so any difference in optimized welfare across degrees or $k$ values is attributable to those parameters and not to sampling variation. The evaluation seed is offset from the training seed by a fixed constant, guaranteeing independence between the two sets.

\subsection{Threshold Parameterization}
\label{appendix:threshold-param}

The participant-audit mechanism uses a polynomial participation threshold of degree $D$,
\begin{equation}
\tau(e;\bsc)=\sum_{m=0}^{D} A_m\, e^m,
\qquad
\bsc=(A_0,A_1,\ldots,A_D)\in\mathbb{R}^{D+1},
\label{eq:tau-poly-appendix}
\end{equation}
where the coefficient vector $\bsc$ is the decision variable optimized by the genetic algorithm. We sweep $D\in\{1,2,3\}$, recovering affine, quadratic, and cubic thresholds.

\subsection{Genetic Algorithm Details}
\label{appendix:ga-details}

The optimization goal is to find
\begin{equation}
\bsc^{*}=\argmax_{\bsc\in\mathcal{C}}\;\bar{W}^{\text{vcgPA}}(\bsc),
\end{equation}
where $\bar{W}^{\text{vcgPA}}(\bsc)$ is the mean participant-audit welfare over the training draws and $\mathcal{C}$ is the gene space defined below. Because the objective is non-convex in $\bsc$ and the feasible space is continuous and high-dimensional, we use a genetic algorithm (GA) operating directly in the original $(v,e)$ space without normalization.

\paragraph{Gene space.}
The gene space $\mathcal{C}=\mathcal{C}_0\times\mathcal{C}_1\times\cdots\times\mathcal{C}_D$ is derived from the empirical distribution of $(v_i,e_i)$ values so that the search range scales correctly regardless of the externality scale $\zeta$ or the polynomial degree $D$. Let $\sigma_v$ denote the standard deviation of agent values across all draws, and let $\overline{|e|^{m}}=\mathbb{E}[|e_i|^{m}]$ denote the mean absolute $m$th power of externalities. The intercept range ensures the threshold can start above or below all bids,
\begin{equation}
\mathcal{C}_0=\Bigl[v_{\min}-\tfrac{1}{2}\sigma_v,\;\; v_{\max}+\tfrac{1}{2}\sigma_v\Bigr],
\label{eq:gene0}
\end{equation}
while for each slope coefficient ($m\geq 1$) the range is calibrated so that at a typical externality the contribution $|A_m e^m|$ can shift $\tau$ by up to $\lambda\sigma_v$, covering the full spread of agent values,
\begin{equation}
\mathcal{C}_m=\Bigl[-\lambda\sigma_v/\overline{|e|^{m}},\;\; \lambda\sigma_v/\overline{|e|^{m}}\Bigr],
\qquad m=1,\ldots,D,
\label{eq:gened}
\end{equation}
where $\lambda=3$ is the gene range multiplier (\cref{tab:ga_hyperparams}).

\paragraph{Vectorized fitness evaluation.}
All $N_{\text{pop}}$ candidate solutions in the current population are evaluated simultaneously using a single matrix multiply per auction draw. Let $\mathbf{P}\in\mathbb{R}^{N_{\text{pop}}\times(D+1)}$ be the population matrix (row $g$ is candidate $\bsc^{(g)}$), and let $\mathbf{E}^{(j)}\in\mathbb{R}^{n\times(D+1)}$ be the Vandermonde-style power matrix for draw $j$ with entries $E^{(j)}_{im}=e_i^{\,m}$. Then
\begin{equation}
\mathbf{T}^{(j)}=\mathbf{P}\bigl(\mathbf{E}^{(j)}\bigr)^{\!\top}\in\mathbb{R}^{N_{\text{pop}}\times n},
\qquad
T^{(j)}_{g,i}=\tau(e_i;\bsc^{(g)}),
\end{equation}
evaluates the threshold for every population member and every agent at once. The admission mask is $A^{(j)}_{g,i}=\mathbf{1}[v_i\geq T^{(j)}_{g,i}]$; agents are sorted by $v$ in descending order, and the cumulative sum of admitted indicators identifies the first $k$ admitted positions, yielding the winner set. Welfare is accumulated by summing $v_i+e_i$ over winners for all $g$ simultaneously. The fitness of candidate $g$ is its mean participant-audit welfare across the $J_{\text{train}}=500$ training draws; the evaluation draws are never seen by the optimizer.

\paragraph{Selection, crossover, and mutation.}
The GA uses tournament selection (for each parent slot, a random subset of the population is drawn without replacement and the highest-fitness candidate wins), single-point crossover (a random cut point in $\{1,\ldots,D\}$ is chosen for each pair of parents, and the offspring inherits coefficients from one parent before the cut and the other after), and Gaussian-perturbation mutation with per-gene step sizes proportional to each gene's range,
\begin{equation}
\sigma_m^{\text{mut}}=\sigma_{\text{frac}}\cdot|\mathcal{C}_m|,
\end{equation}
where $|\mathcal{C}_m|$ is the width of gene range $\mathcal{C}_m$. Each gene is independently perturbed with probability $p_{\text{mut}}$; to prevent stagnation, at least one gene per offspring is always mutated, and all perturbed values are clipped to the gene bounds.

\paragraph{Multi-restart and warm start.}
To increase coverage of the coefficient space, each cell is run for $N_{\text{restart}}$ independent restarts with different random seeds, and the highest-fitness solution is retained. When sweeping degrees $D=1,2,3$ in ascending order, the optimal solution from degree $D$ seeds the first population member at degree $D+1$,
\begin{equation}
\bsc_{\text{seed}}^{(D+1)}=\bigl(A_0^{*},A_1^{*},\ldots,A_D^{*},\,0\bigr),
\label{eq:warm-start}
\end{equation}
so that a higher-degree search can never perform worse than the best lower-degree solution it nests. The warm start is applied only to restart~0; remaining restarts use fresh random populations, preserving the benefit of the warm start while maintaining diversity.

\Cref{tab:ga_hyperparams} reports the default hyperparameters used across experiments. Full implementation details and the optimization code itself are available in the released supplemental materials (Appendix~\ref{appendix:git-repo}).

\begin{table}[H]
\centering
\caption{Genetic algorithm hyperparameters (defaults used unless noted otherwise).}
\label{tab:ga_hyperparams}
\begin{tabularx}{\textwidth}{llX}
\toprule
Symbol & Default & Description \\
\midrule
$D$ & 1, 2, 3 & Polynomial degree of $\tau$ (swept) \\
$N_{\text{pop}}$ & 50 & Population size \\
$T$ & 500 & GA generations per restart \\
$N_{\text{restart}}$ & 1 & Independent restarts per cell \\
$N_{\text{mate}}$ & 15 & Parents selected for crossover each generation \\
$p_{\text{mut}}$ & 0.25 & Per-gene mutation probability \\
$\sigma_{\text{frac}}$ & 0.15 & Gaussian mutation step as fraction of gene range \\
$\lambda$ & 3.0 & Gene range multiplier (Eq.~\ref{eq:gened}) \\
$J_{\text{train}}$ & 500  & Training auction draws per cell (GA fitness) \\
$J_{\text{eval}}$  & 2000 & Evaluation draws per cell (welfare reporting only) \\
$n$ & 8, 16 & Bidders sampled per auction draw (swept) \\
\bottomrule
\end{tabularx}
\end{table}

\FloatBarrier

\subsection{Expected Penalty Curves}
\label{appendix:penalty-curves}

\Cref{sec:participation-threshold-analysis} shows that, in the i.i.d.\ setting, a participation threshold $\tau$ and a marginal value distribution $F$ together determine the participant penalty function $r$ via
$r(e)=\mathbb{E}[x_i(\tau(e),v_{-i})]\,\tau(e)-\mathbb{E}[p_i(\tau(e),v_{-i})]$.
Given the optimized threshold $\tau(e;\bsc^{*})$, we estimate this penalty empirically from the evaluation draws, producing the participant and recipient penalty curves plotted in the threshold-and-penalty figures.

\paragraph{Participant audit.}
Let $\{p_j\}_{j=1}^{J_{\text{eval}}}$ be the sequence of participant-audit auction payments (the $(k{+}1)$th-price payments among admitted agents defined in \cref{sec:numerical-empirical}), for each evaluation draw. For an agent with externality $e$, define
\begin{align}
A(e)&=\frac{1}{J_{\text{eval}}}\bigl|\{\,j:p_j<\tau(e;\bsc^{*})\,\}\bigr|,\\
P(e)&=\mathbb{E}\!\left[p_j\;\middle|\;p_j<\tau(e;\bsc^{*})\right],
\end{align}
so that $A(e)$ is the fraction of auction draws in which the market price falls below a threshold value $\tau(e)$ (i.e.\ the probability that an agent whose value sits exactly at the threshold $\tau(e)$ is allocated), and $P(e)$ is the mean of the prices below that threshold value. To induce the threshold $\tau(e)$, the participant penalty must equal the expected allocation utility of a marginal agent at type $(e,\tau(e))$, giving the empirical penalty
\begin{equation}
r^{\text{PA}}(e)=A(e)\,\bigl[\tau(e;\bsc^{*})-P(e)\bigr].
\label{eq:r-pa}
\end{equation}
For large negative externalities, $\tau(e;\bsc^{*})$ is high relative to the market price, so both $A(e)$ and the expected gap are large and the penalty is steep; for small or positive externalities, $\tau(e;\bsc^{*})$ is low and most market prices clear it, so $A(e)\approx 0$ and $r^{\text{PA}}(e)\approx 0$.

\paragraph{Recipient audit.}
For the recipient audit, selection is based directly on social value $w_i=v_i+e_i$, since the welfare-maximizing recipient penalty from \cref{thm:welfare-maximal-recipient-audit} fully internalizes externality,
\begin{equation}
R^{\text{RA}}(e)=-e.
\label{eq:r-ra}
\end{equation}
This is linear in $e$: agents generating negative externalities ($e<0$) pay a penalty equal to $|e|$, while those generating positive externalities ($e>0$) effectively receive a subsidy equal to $|e|$.

\FloatBarrier

\section{Additional Robustness Figures for Simulated Data}
\label{appendix:robustness-figures}

This appendix reports additional robustness analyses and sensitivity experiments on the synthetic datasets described in \cref{appendix:simulated-data-text}, supplementing the representative results presented in the main text. Even more results can be found through the online code repository linked in Appendix \ref{appendix:git-repo}.

\subsection{Threshold Sensitivity to Competition}
\label{appendix:optimization-figures}

Figure~\ref{fig:tau_by_n_sim_appendix} shows how optimized participation thresholds vary with the number of auction participants $n$. As competition increases, the mechanism becomes more selective and participation thresholds shift toward excluding larger portions of low-welfare bidders.

\begin{figure}[ht]
    \centering
    \includegraphics[width=0.8\linewidth]{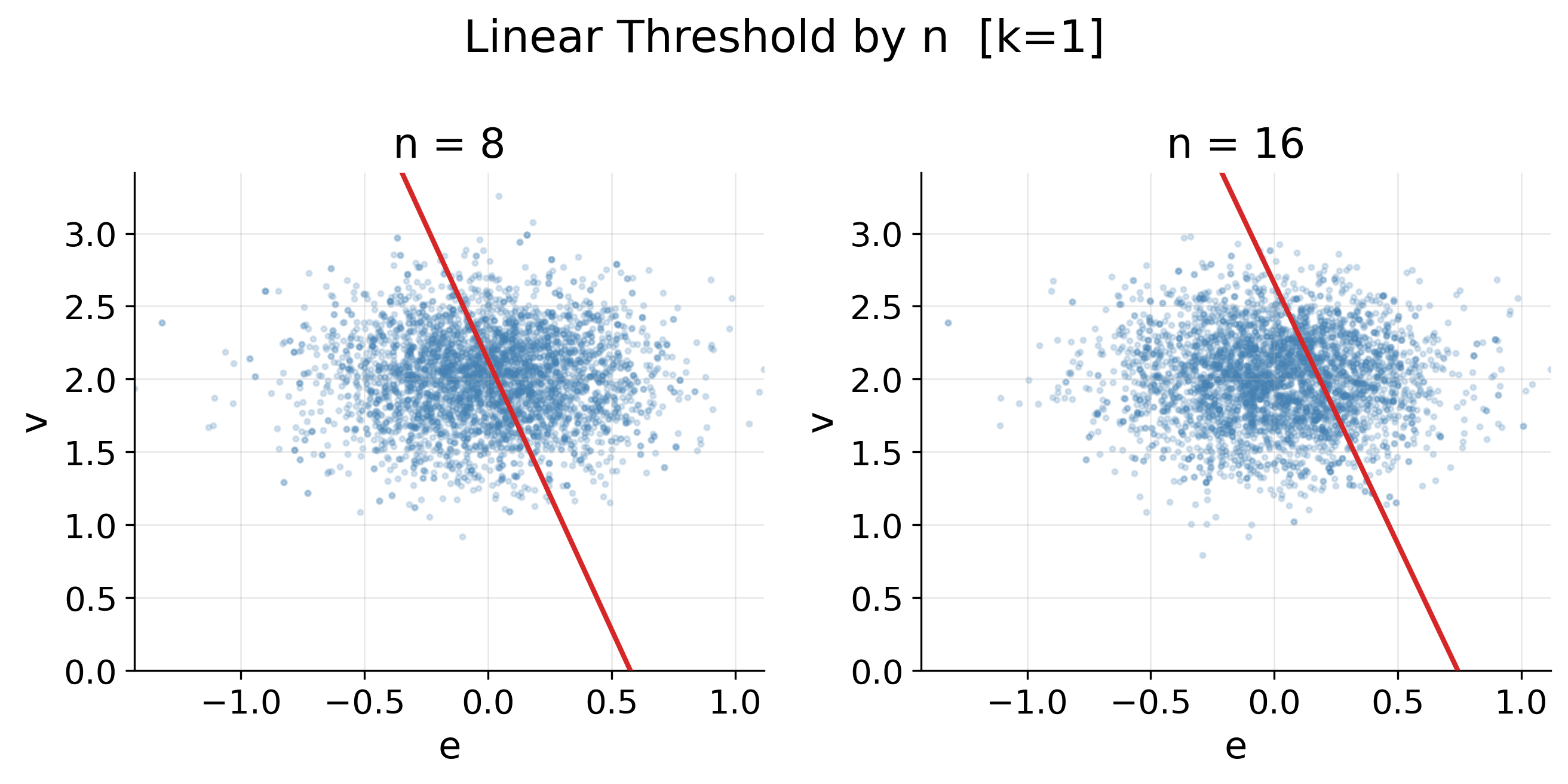}
    \caption{Optimized participant threshold functions for varying numbers of auction participants.}
    \label{fig:tau_by_n_sim_appendix}
\end{figure}

Figure~\ref{fig:tau_by_k_sim_appendix} shows how optimized thresholds vary with the number of allocated items $k$. Increasing the number of allocated items relaxes participation constraints because the mechanism can simultaneously accommodate bidders with differing valuation--externality tradeoffs.

\begin{figure}[ht]
    \centering
    \includegraphics[width=0.8\linewidth]{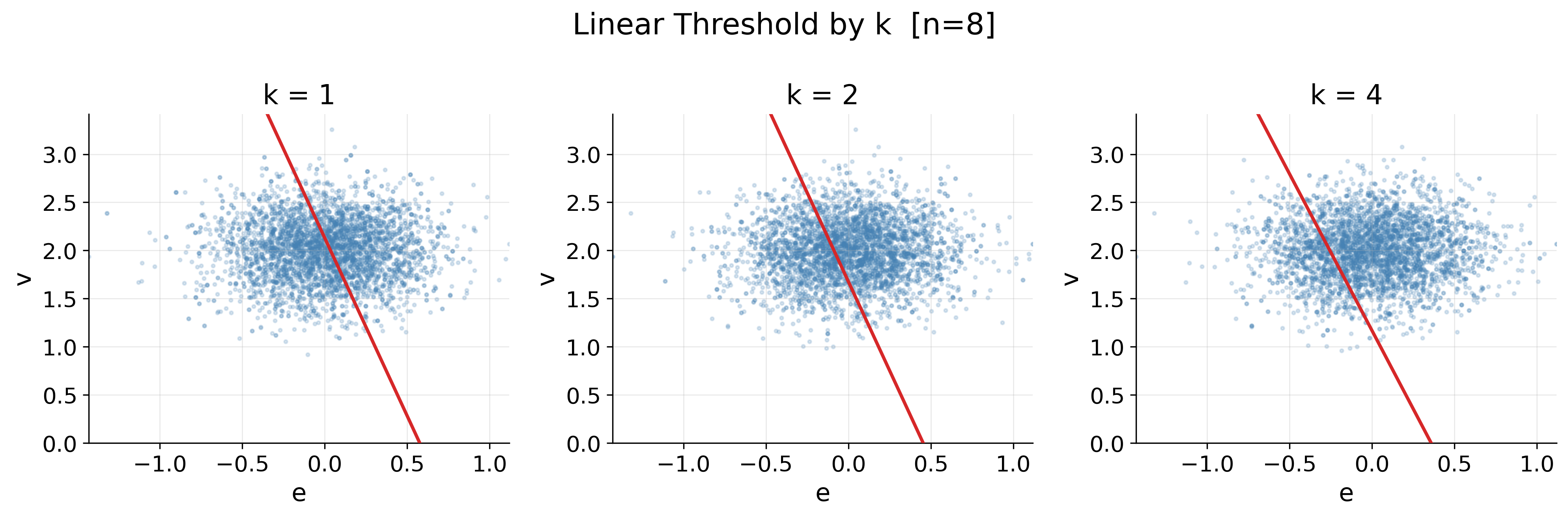}
    \caption{Optimized participant threshold functions for varying numbers of allocated items.}
    \label{fig:tau_by_k_sim_appendix}
\end{figure}

Figure~\ref{fig:tau_by_degree_sim_appendix} visualizes threshold functions explored during the genetic optimization procedure. The highlighted curve corresponds to the best-performing threshold function identified by the optimization algorithm.

\begin{figure}[ht]
    \centering
    \includegraphics[width=0.8\linewidth]{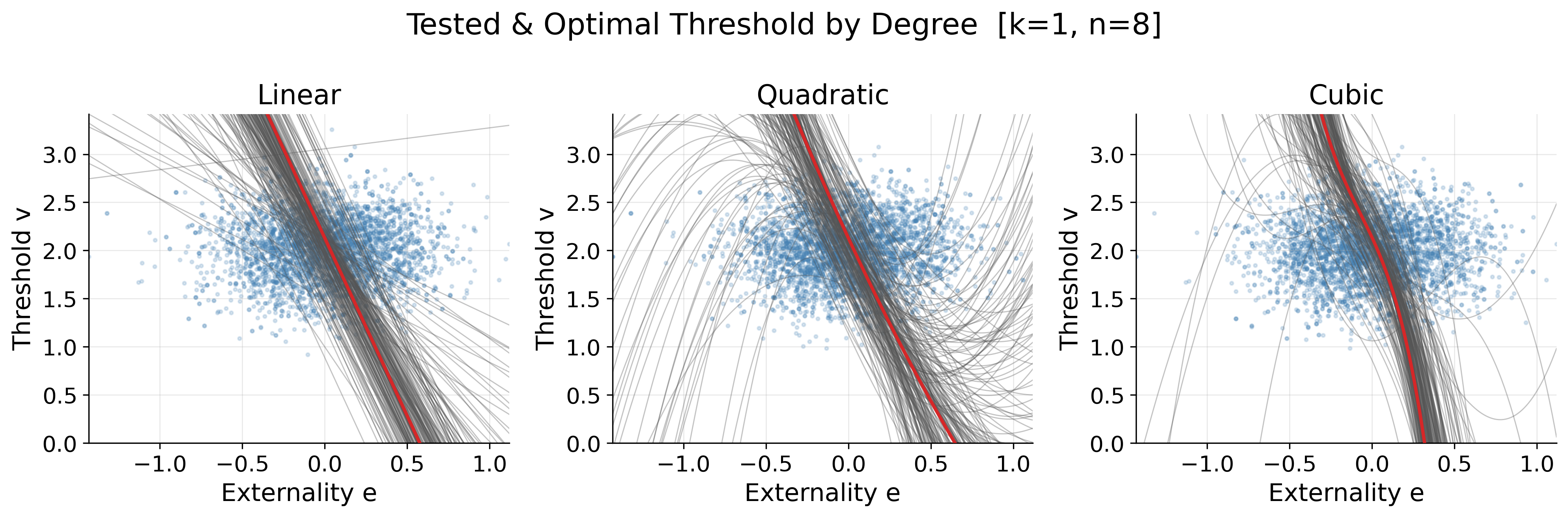}
    \caption{Threshold functions explored during genetic optimization, with the best threshold found highlighted.}
    \label{fig:tau_by_degree_sim_appendix}
\end{figure}

\subsection{Threshold and Penalty Functions Across Bimodal Distributions}

\begin{figure}
    \centering
    \includegraphics[width=\linewidth]{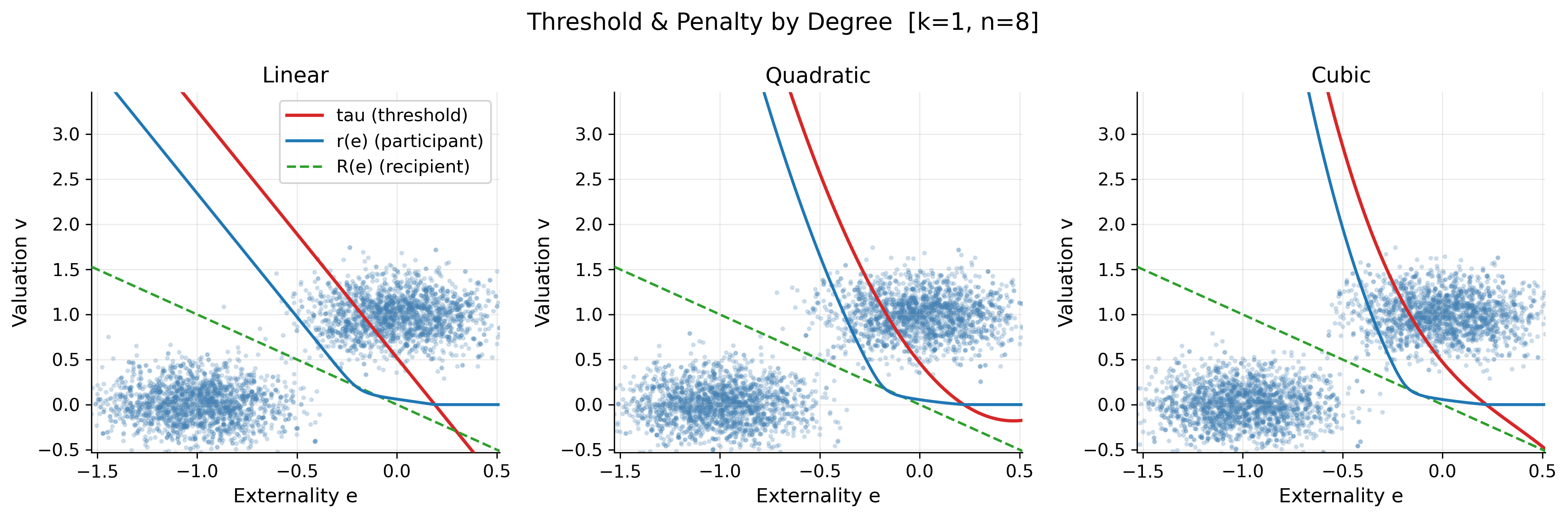}
    \caption{Threshold and penalty functions for diagonal bi-modal normally distributed synthetic data with $k=1$ and $n=8$.}
    \label{fig:c_penalty_by_degree_k1_n8}
\end{figure}

\begin{figure}
    \centering
    \includegraphics[width=\linewidth]{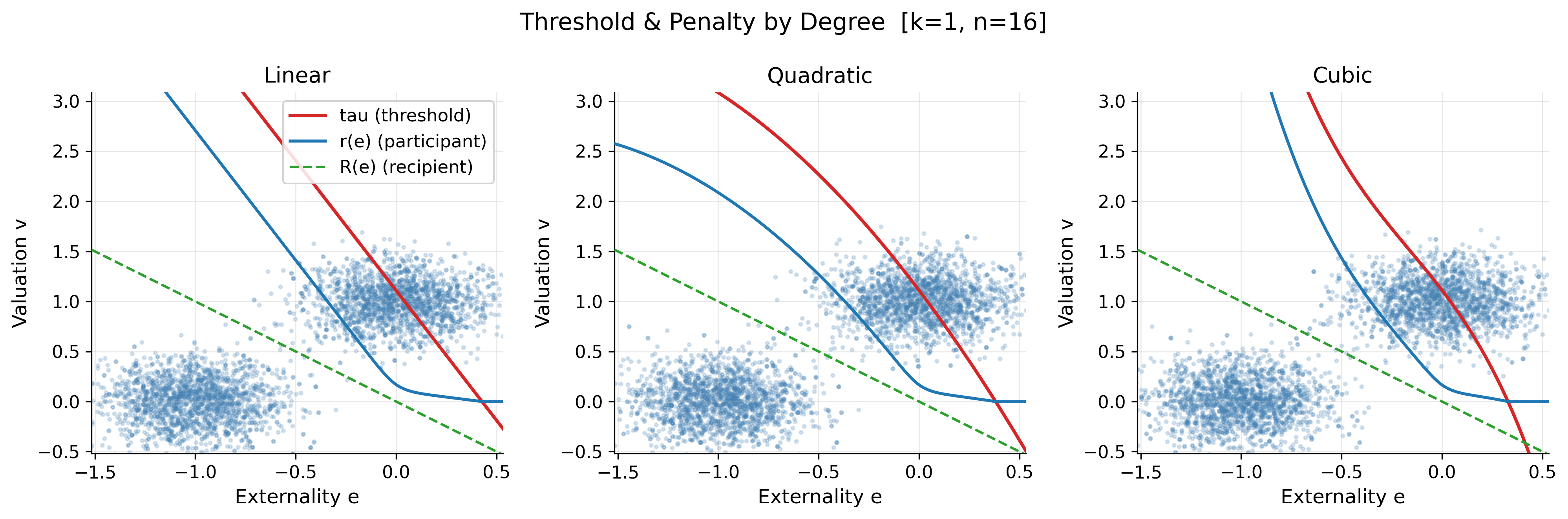}
    \caption{Threshold and penalty functions for diagonal bi-modal normally distributed synthetic data with $k=1$ and $n=16$.}
    \label{fig:c_penalty_by_degree_k1_n16}
\end{figure}

\begin{figure}
    \centering
    \includegraphics[width=\linewidth]{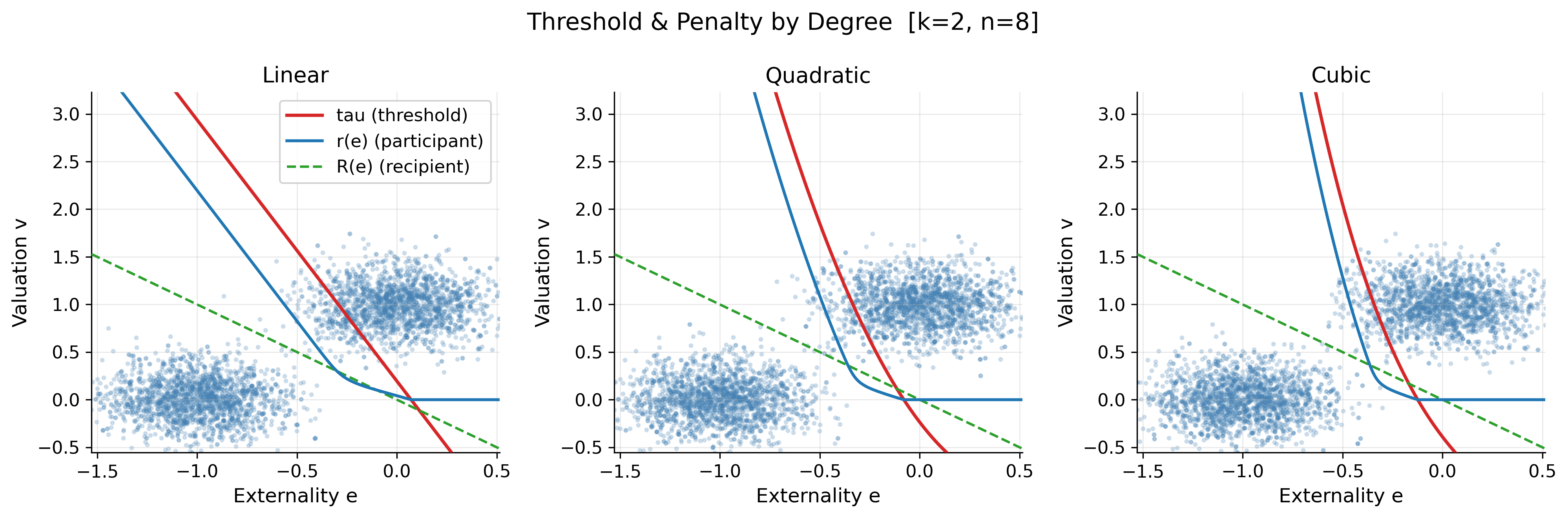}
    \caption{Threshold and penalty functions for diagonal bi-modal normally distributed synthetic data with $k=2$ and $n=8$.}
    \label{fig:c_penalty_by_degree_k2_n8}
\end{figure}

\begin{figure}
    \centering
    \includegraphics[width=\linewidth]{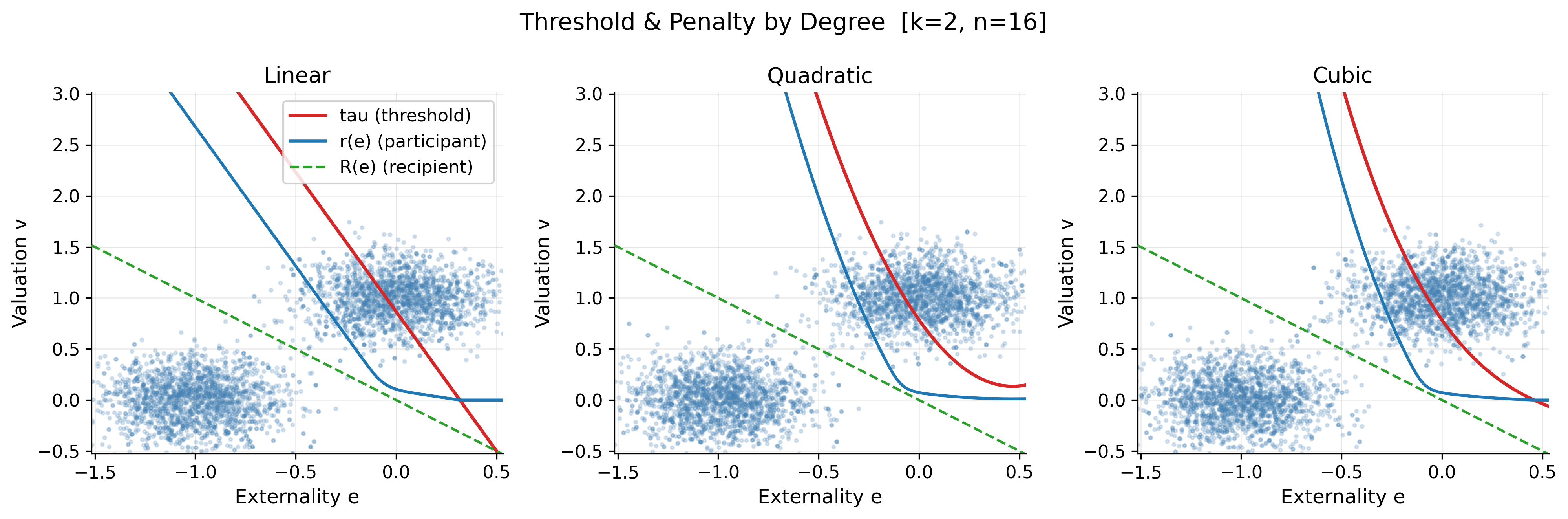}
    \caption{Threshold and penalty functions for diagonal bi-modal normally distributed synthetic data with $k=2$ and $n=16$.}
    \label{fig:c_penalty_by_degree_k2_n16}
\end{figure}

\begin{figure}
    \centering
    \includegraphics[width=\linewidth]{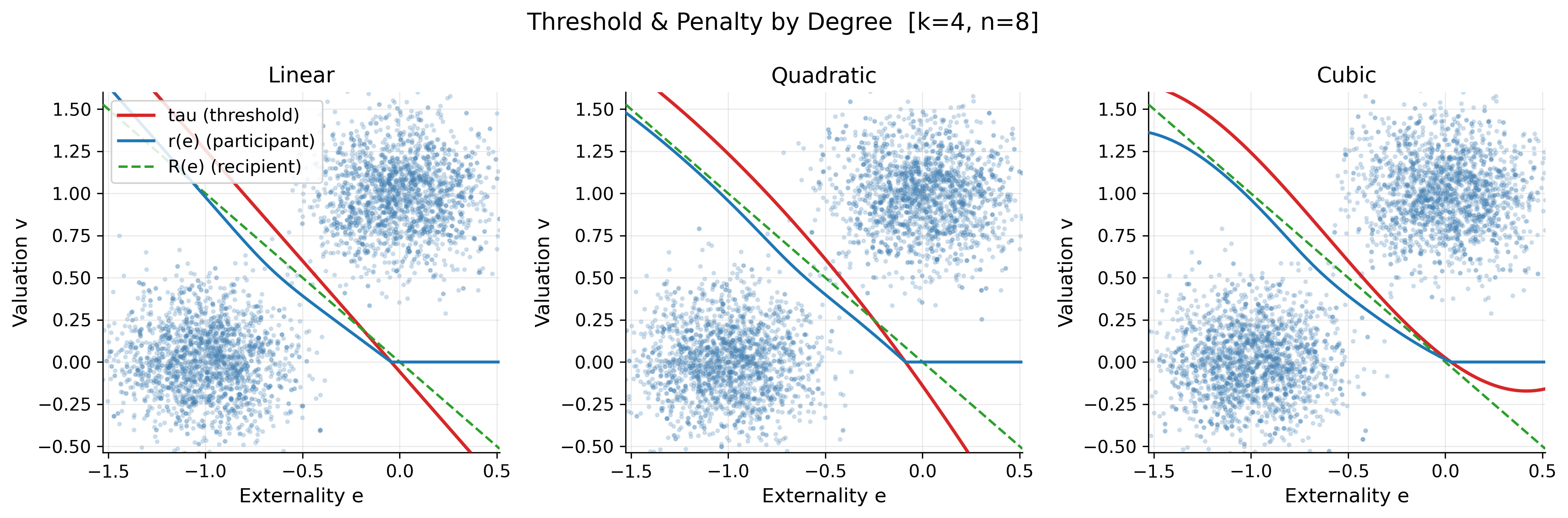}
    \caption{Threshold and penalty functions for diagonal bi-modal normally distributed synthetic data with $k=4$ and $n=8$.}
    \label{fig:c_penalty_by_degree_k4_n8}
\end{figure}

\begin{figure}
    \centering
    \includegraphics[width=\linewidth]{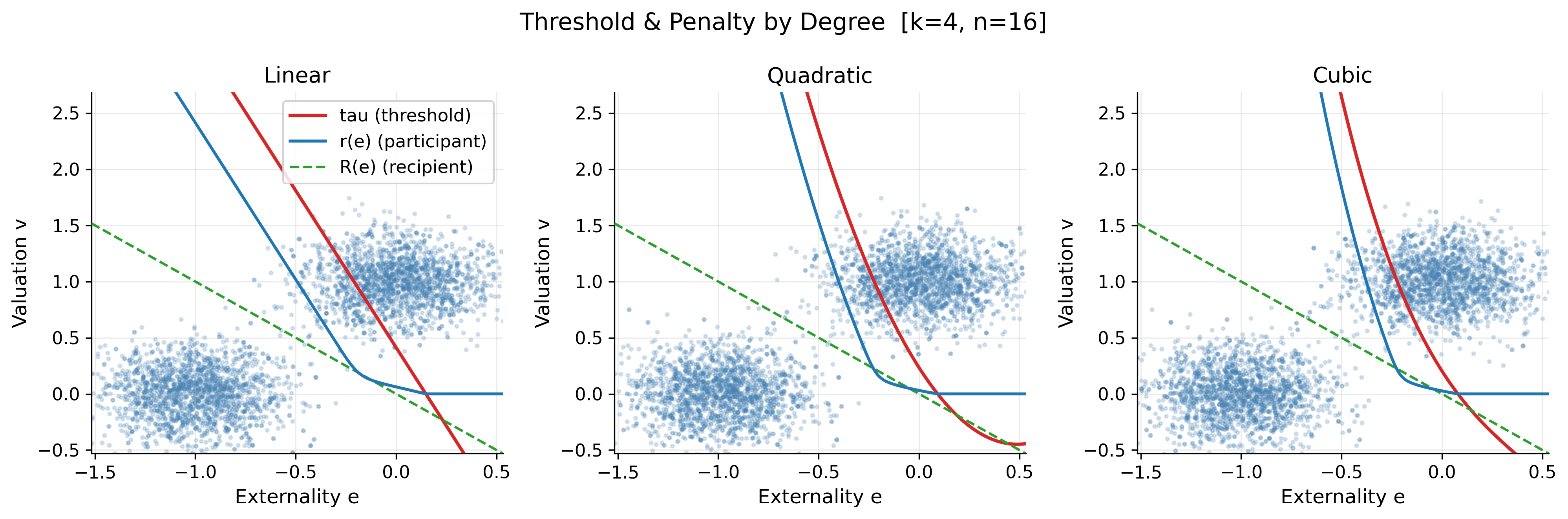}
    \caption{Threshold and penalty functions for diagonal bi-modal normally distributed synthetic data with $k=4$ and $n=16$.}
    \label{fig:c_penalty_by_degree_k4_n16}
\end{figure}

\begin{figure}
    \centering
    \includegraphics[width=\linewidth]{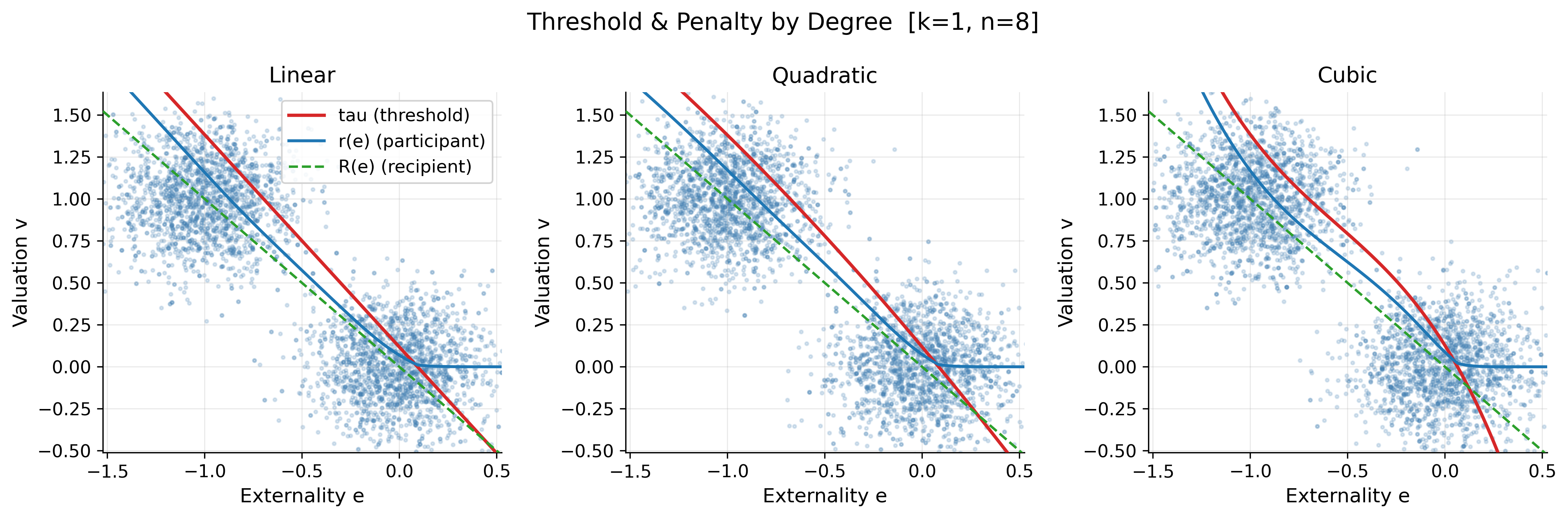}
    \caption{Threshold and penalty functions for diagonal bi-modal normally distributed synthetic data with $k=1$ and $n=8$.}
    \label{fig:d_penalty_by_degree_k1_n8}
\end{figure}

\begin{figure}
    \centering
    \includegraphics[width=\linewidth]{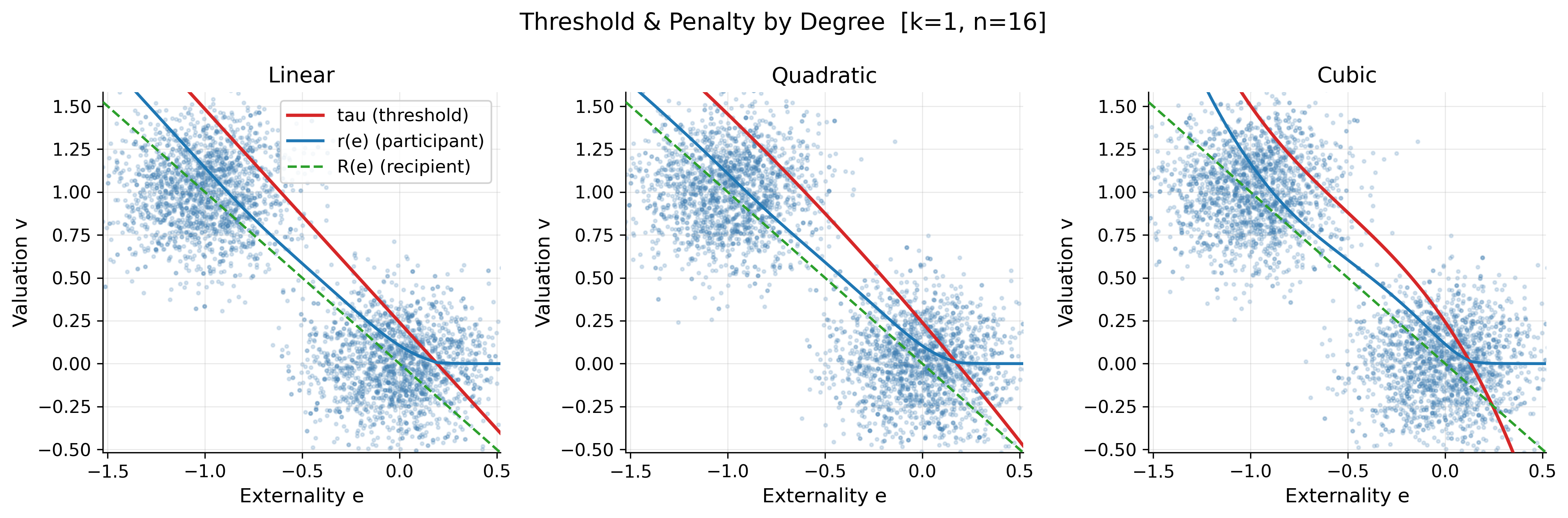}
    \caption{Threshold and penalty functions for diagonal bi-modal normally distributed synthetic data with $k=1$ and $n=16$.}
    \label{fig:d_penalty_by_degree_k1_n16}
\end{figure}

\begin{figure}
    \centering
    \includegraphics[width=\linewidth]{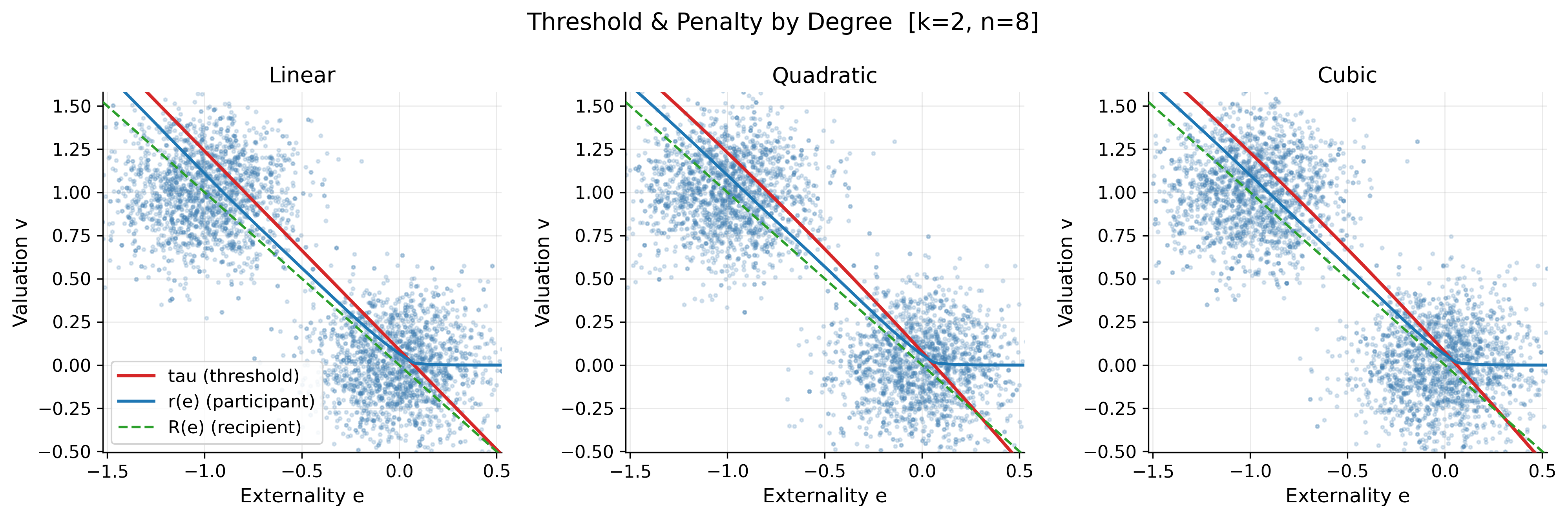}
    \caption{Threshold and penalty functions for diagonal bi-modal normally distributed synthetic data with $k=2$ and $n=8$.}
    \label{fig:d_penalty_by_degree_k2_n8}
\end{figure}

\begin{figure}
    \centering
    \includegraphics[width=\linewidth]{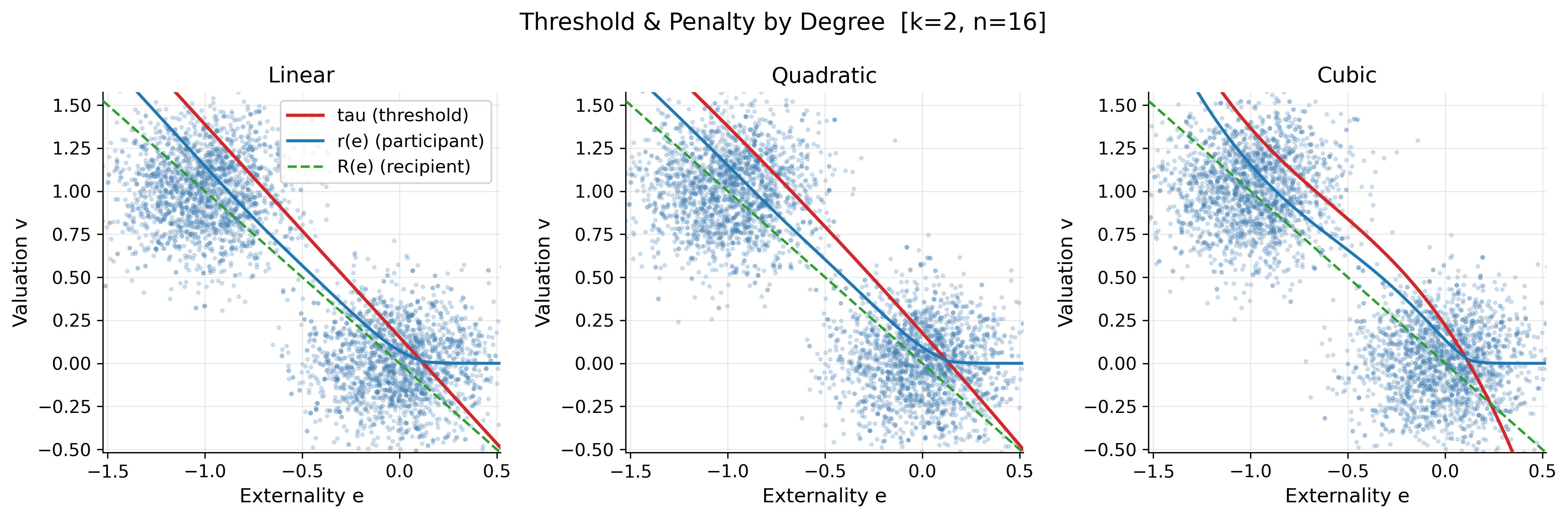}
    \caption{Threshold and penalty functions for diagonal bi-modal normally distributed synthetic data with $k=2$ and $n=16$.}
    \label{fig:d_penalty_by_degree_k2_n16}
\end{figure}

\begin{figure}
    \centering
    \includegraphics[width=\linewidth]{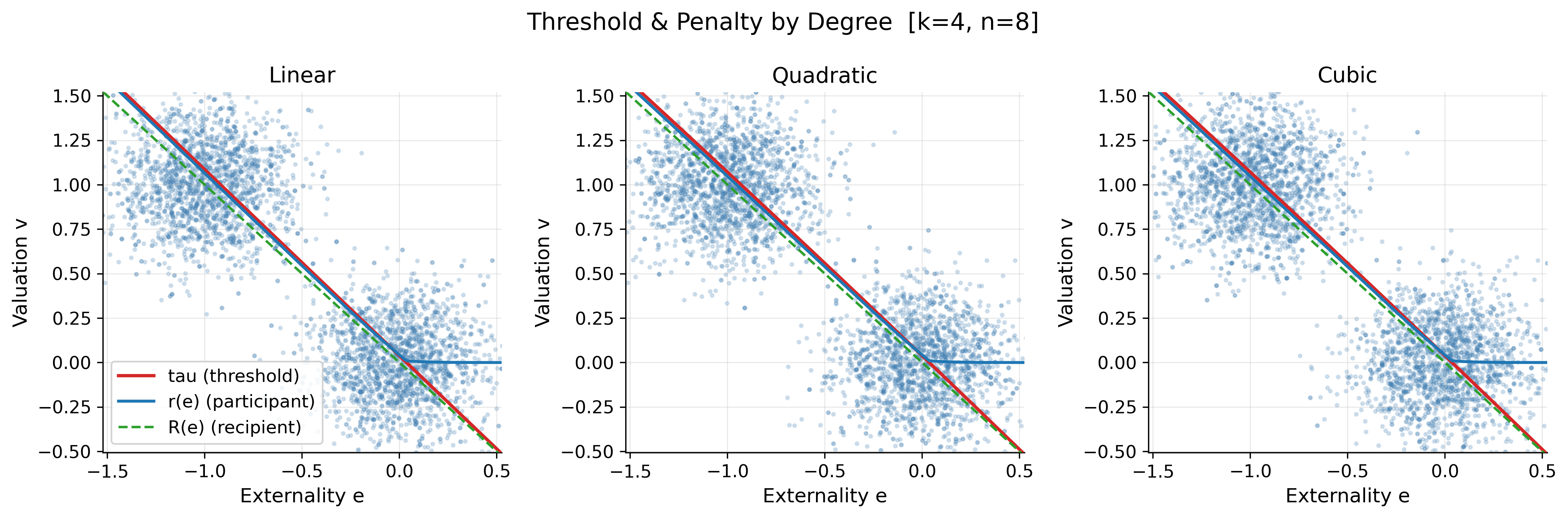}
    \caption{Threshold and penalty functions for diagonal bi-modal normally distributed synthetic data with $k=4$ and $n=8$.}
    \label{fig:d_penalty_by_degree_k4_n8}
\end{figure}

\begin{figure}
    \centering
    \includegraphics[width=\linewidth]{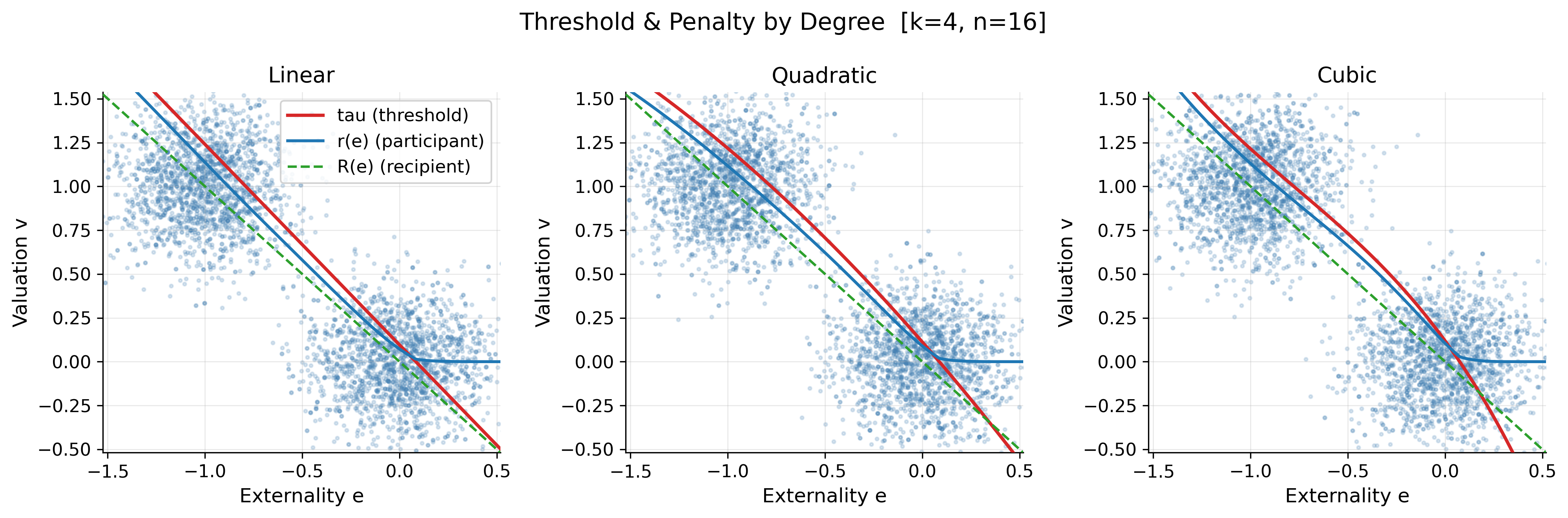}
    \caption{Threshold and penalty functions for diagonal bi-modal normally distributed synthetic data with $k=4$ and $n=16$.}
    \label{fig:d_penalty_by_degree_k4_n16}
\end{figure}

\FloatBarrier

\section{X Community Notes Dataset Construction}
\label{appendix:dataset}

\subsection{Dataset Collection and Processing}
\label{appendix:data-collection}

We construct an empirical dataset of X-Twitter posts annotated with Community Notes (which we refer to as the XNP400 dataset) using publicly available Community Notes data together with post-level metadata retrieved from the X API, as of February~2,~2025. Community Notes are labels which any User on X may choose to apply to a post, regarding how misleading the post is. Other users may also rate the helpfulness of any Note.

To collect these posts, we used a list of community notes downloaded from the \href{https://communitynotes.x.com/guide/en/under-the-hood/download-data}{X community note data library} in February, 2025. We then filtered the list of notes to ones with at least 400 ratings. Finally, we download posts linked to the remaining community notes using the X API and remove outlier posts by rating and action count. Therefore, posts in the XNP400 dataset have the following attributes: (i) the post has at least one Community Note, since posts without any note carry no externality signal; (ii) the post has positive recorded lifetime impressions, since per-month scores are otherwise undefined; (iii) the post is not among the top or bottom $0.1\%$ of posts by aggregate post rating $\beta_i$ or by total action count $\alpha_i$, trimming approximately 24 posts at each tail to stabilize the genetic optimization; and (iv) the post's agreement score---aggregate rating per 1{,}000 impressions---is not below $-400$, since we believe such posts might reflect disingenuous coordinated ratings rather than genuine user experiences. After applying all filters, the resulting dataset contains $N=11{,}853$ posts.

\subsection{Externality Estimation}
\label{appendix:externality-estimation}

For each post $i$, we compute an externality score using Community Notes and ratings using the following method. Each helpfulness rating is mapped to a numerical helpfulness score as follows, where the $l$'th rating on the $j$'th note on the $i$'th post is denoted by $\rho_{ijl}$.
\[
\rho_{ijl} =
\begin{cases}
1   & \texttt{HELPFUL} \\
0.5 & \texttt{SOMEWHAT\_HELPFUL} \\
0   & \texttt{NOT\_HELPFUL}.
\end{cases}
\]

We apply a misleading score of $-1$ to all Notes labeling the post as some form of misleading, and a misleading score of $+1$ to all Notes suggesting some form of not misleading, where $\pi_{ij}$ denotes the misleading score of Note $j$ on Post $i$.
\[
\pi_{ij} =
\begin{cases}
-1 & \texttt{MISINFORMED\_OR\_POTENTIALLY\_MISLEADING} \\
+1 & \texttt{NOT\_MISLEADING}.
\end{cases}
\]

We denote post $i$'s externality score with $\beta_i$, defined as follows, where $\mathcal{J}(i)$ denote the set of Community Notes attached to post $i$, and let $\mathcal{L}(i,j)$ denote the set of ratings associated with note $j\in\mathcal{J}(i)$.
\[
\beta_i
=
\sum_{j\in\mathcal{J}(i)}
\pi_{ij}
\left(
1+\sum_{l\in\mathcal{L}(i,j)} \rho_{ijl}
\right)
\]
This method simply counts the number of ratings in support of \texttt{Not Misleading} and \texttt{Misleading}, and subtracts the \texttt{Misleading} total from the \texttt{Not Misleading} total. We treat the note author as an implicit supportive rater. Therefore, a Post with overwhelming support for \texttt{Misleading} Notes results in a strongly negative values of $\beta_i$. Thus, we are assuming there is a linear relationship between the occurrences of Community Note ratings and monetary values of externality from posts.

Finally, to convert the externality score to an externality in dollars per month, we divide by the age of the post in months denoted by $\mu_i$, and multiply by an `externality magnitude' scaling parameter with units dollars per externality-score (also interpretable as per rating) denoted by $\zeta$. Thus, the estimated externality of post $i$ is
\[
e_i
=
\zeta \cdot \frac{\beta_i}{\mu_i}.
\]

Because this model of externalities is crude, and the true magnitude of externalities relative to valuations is uncertain, we calibrate $\zeta$ against external estimates and simulate audited auctions across a range of $\zeta$ values; see \Cref{sec:welfare-comparisons}.

\subsection{Valuation Estimation}
\label{appendix:valuation-estimation}

To approximate bidder valuations, we use post engagement activity as a proxy for advertiser value. On X-Twitter, promoted posts may be priced according to impressions or user actions such as likes, replies, reposts, and clicks \cite{x_corporation_x_2025, x_corporation_ads_2025, x_corporation_bidding_2025}. Consequently, posts generating greater engagement (actions per impression) plausibly correspond to higher advertiser valuations. Once again, this method should only be interpreted to provide a relative ranking and rough estimate of advertiser valuations.

Let $\alpha_i$ denote the total number of actions on post $i$, and let $\mu_i$ denote the age of the post in months. We estimate the valuation of post $i$ as
\[
v_i
=
\theta \cdot \frac{\alpha_i}{\mu_i},
\]
where $\theta\ge0$ is a scaling parameter representing approximate dollar value per action. Thus we are also assuming a linear relationship between actions per month and valuation. Estimates of advertiser costs on X-Twitter vary substantially across campaigns and targeting strategies, however, we consider an average valuation of \$1/action to be a reasonable rough estimate \cite{x_corporation_bidding_2025, webfx_how_2025}. So, for simplicity we fix $\theta = 1$ and vary $\zeta$. This also allows valuations and externality to be interpreted in units relative to each other rather than absolute units.

\subsection{Calibration of the Externality Scale \texorpdfstring{$\zeta$}{zeta}}
\label{appendix:zeta-calibration}

Because the externality scale $\zeta$ converts the dimensionless aggregate rating $\beta_i$ into a dollar externality, we anchor it to two published estimates of the dollar value of social-media content externalities. Each estimate is expressed per impression or per action, so we convert it into a cost per unit of aggregate rating by dividing by the corresponding sample mean from our dataset.

\paragraph{Impression-based estimate.}
Goldstein et al.~\cite{goldstein_economic_2014} estimate a cost of approximately \$1.53 per thousand spam impressions (a CPM of \$1.53). Dividing the per-impression cost by the sample mean aggregate rating per impression, $\overline{\beta_i/\eta_i}$ (where $\eta_i$ is post $i$'s lifetime impressions), gives
\begin{equation}
\zeta^{\text{imp}}=\frac{1.53/1000}{\;\overline{\beta_i/\eta_i}\;}\approx \$2.71 \text{ per rating unit per month.}
\end{equation}

\paragraph{Action-based estimate.}
Schnadower Mustri et al.~\cite{schnadower_mustri_behavioral_2023} estimate a cost of \$0.54 per interaction with a targeted advertisement. Dividing by the sample mean aggregate rating per action, $\overline{\beta_i/\alpha_i}$, gives
\begin{equation}
\zeta^{\text{act}}=\frac{0.54}{\;\overline{\beta_i/\alpha_i}\;}\approx \$2.54 \text{ per rating unit per month.}
\end{equation}

The two estimates agree closely, yielding a baseline calibration of $\zeta\approx\$2.5$; the representative grid value $\zeta=2.8$ used in the main-text sweep sits close to both. Because the underlying mapping from Community Note activity to monetary externalities is inherently uncertain, we treat $\zeta$ as a free parameter and report results across the range $\zeta\in\{1,2.8,4.6,6.4,8.2,10\}$.

\subsection{Descriptive Statistics}
\label{appendix:descriptive-statistics}

Table~\ref{tab:tweet_summary_stats} reports descriptive statistics for the post-level metrics used in constructing the empirical valuation and externality estimates.
\begin{table}[htbp]
\centering
\caption{Summary statistics for post-level metrics.  ($N=11{,}811$), \cref{appendix:data-collection}).}
\label{tab:tweet_summary_stats}
\small
\resizebox{\textwidth}{!}{%
\begin{tabular}{lrrrrrrr}
\toprule
Metric & Mean & Std.\ Dev. & Min & p25 & Median & p75 & Max \\
\midrule
Lifetime Impressions & 13,550,715 & 23,833,062 & 7,703.0 & 1,904,286 & 5,087,266 & 14,635,842 & 634,387,512 \\
Impressions per Month & 3,382,788 & 9,950,972 & 1,135.8 & 233,377.5 & 758,325.6 & 2,565,235 & 358,436,171 \\
Lifetime Engagements ($\alpha_i$) & 70,824.7 & 135,363.4 & 150 & 5,880.5 & 20,221.0 & 70,909.0 & 1,519,232 \\
Aggregate Post Rating ($\beta_i$) & -892 & 1,338.4 & -13,331.5 & -1,210.2 & -746 & -475 & 9,319.5 \\
Aggregate Rating per Impression & -4.75e{-}4 & 0.00157 & -0.0705 & -5.01e{-}4 & -1.78e{-}4 & -4.92e{-}5 & 0.0216 \\
Aggregate Rating per Month & -200 & 647 & -21,341.0 & -222 & -93.7 & -42.6 & 10,745.0 \\
Value Score ($v_i$) & 19,764.1 & 61,655.1 & 8.73 & 662 & 2,823.3 & 13,196.4 & 1,611,199 \\
Externality Score ($e_i$, $\zeta{=}1$) & -200 & 647 & -21,341.0 & -222 & -93.7 & -42.6 & 10,745.0 \\
\bottomrule
\end{tabular}}

\vspace{2pt}
{\small\textit{Note:} Rate variables are computed by dividing lifetime totals by post age in months as of February~2,~2025.}
\end{table}

Table~\ref{tab:tweet_correlations} reports pairwise Pearson correlations among several constructed post-level metrics.

\begin{table}[htbp]
\centering
\caption{Pairwise Pearson correlations between post-level metrics.}
\label{tab:tweet_correlations}
\small
\begin{tabular}{lrrrrr}
\toprule
Metric & (1) & (2) & (3) & (4) & (5) \\
\midrule
(1)~Impressions per Month    & 1.00  &        &  &  &  \\
(2)~$\beta$ per Impression  & 0.09  & 1.00   &  &  &  \\
(3)~$\alpha$ per Impression & -0.05 & -0.18  & 1.00 &  &  \\
(4)~$\beta$ per Month       & 0.06  & 0.08   & 0.01 & 1.00 &  \\
(5)~$\alpha$ per Month      & 0.72  & 0.08   & 0.20 & 0.08 & 1.00 \\
\bottomrule
\end{tabular}

\vspace{2pt}
{\small\textit{Note:} Lower triangle only; diagonal entries equal 1 by construction.}
\end{table}

Figure~\ref{fig:post_note_dist} reports distributions of the number of Notes on X Posts, the type of Notes, and the type of Ratings on Notes.

\begin{figure}[ht]
    \centering
    \includegraphics[width=\linewidth]{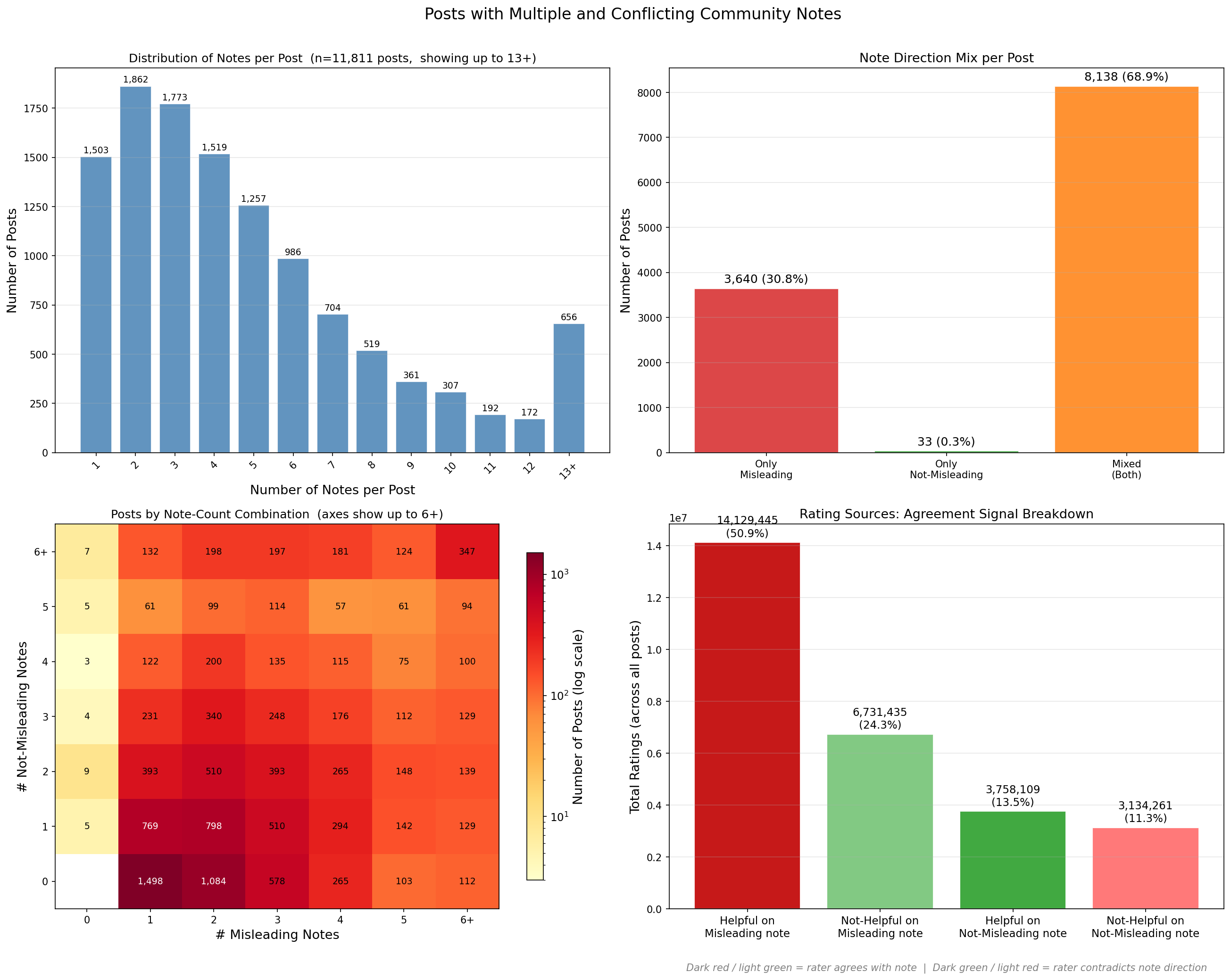}
    \caption{Distributions of Notes and Ratings on X Posts}
    \label{fig:post_note_dist}
\end{figure}

Figure~\ref{fig:empirical_dist} displays the distribution of unadjusted externality and valuation scores in the XNP400 dataset.

\begin{figure}[ht]
    \centering
    \includegraphics[width=0.75\linewidth]{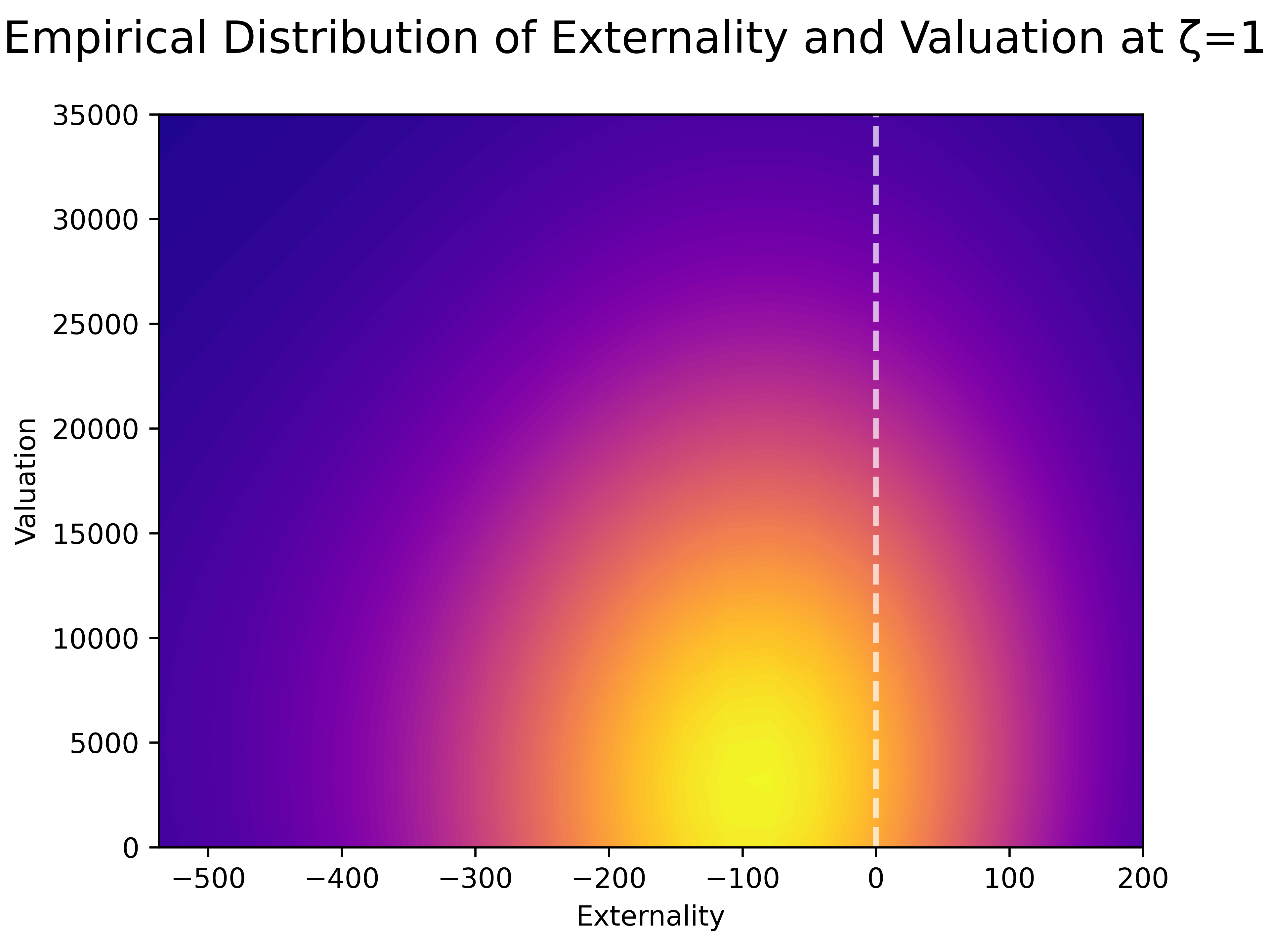}
    \caption{Distribution of externality and valuation scores on X Posts}
    \label{fig:empirical_dist}
\end{figure}

\FloatBarrier

\section{Online Supplemental Materials}
\label{appendix:git-repo}

Code and other online supplemental materials can be found at

\noindent \href{https://github.com/gabemgem/audited-auctions-supplemental}{https://github.com/gabemgem/audited-auctions-supplemental}.

\end{document}